\documentclass[a4paper,12pt,english]{scrartcl}

\usepackage[ansinew]{inputenc}
\usepackage{amsmath}
\usepackage{amsthm}
\usepackage{amsfonts}
\usepackage{amssymb}
\usepackage{array} 
\usepackage{graphicx}
\usepackage{float}
\usepackage[dvipsnames,table]{xcolor} 
\usepackage[colorlinks=true,linkcolor=RoyalBlue,citecolor=PineGreen,urlcolor=RoyalBlue]{hyperref} 

\newtheoremstyle{mytheorem}{}{}{\itshape}{}{\bfseries}{}{\newline}{\thmname{#1}\thmnumber{ #2.}\thmnote{ \textup{(#3)}}}
\theoremstyle{mytheorem}
\newtheorem{theorem}{Theorem}[section]
\newtheorem{definition}[theorem]{Definition}
\newtheorem{lemma}[theorem]{Lemma}
\newtheorem{corollary}[theorem]{Corollary}
\newtheorem{proposition}[theorem]{Proposition}
\newtheorem{example}[theorem]{Example}
\newtheorem{remark}[theorem]{Remark}
\newtheorem{problem}[theorem]{Problem}
\newtheorem{algo}[theorem]{Algorithm}

\newcommand{\A}{\mathbb A} 
\newcommand{\N}{\mathbb N} 
\newcommand{\Z}{\mathbb Z} 

\newcommand{\K}{\mathbb K} 
\newcommand{\ie}{{\it i.e.}}

\newcommand{\E}{\mathbb E} 
\newcommand{\AODE}{{AO{$\Delta$}E}}
\newcommand{\AODEs}{{AO{$\Delta$}Es}}

\providecommand{\keywords}[1]
{
  \small	
  \textbf{\textit{Keywords---}} #1
}

\DeclareMathOperator{\Res}{Res}

\makeatletter
\DeclareOldFontCommand{\rm}{\normalfont\rmfamily}{\mathrm}
\DeclareOldFontCommand{\sf}{\normalfont\sffamily}{\mathsf}
\DeclareOldFontCommand{\tt}{\normalfont\ttfamily}{\mathtt}
\DeclareOldFontCommand{\bf}{\normalfont\bfseries}{\mathbf}
\DeclareOldFontCommand{\it}{\normalfont\itshape}{\mathit}
\DeclareOldFontCommand{\sl}{\normalfont\slshape}{\@nomath\sl}
\DeclareOldFontCommand{\sc}{\normalfont\scshape}{\@nomath\sc}
\makeatother

\title{Rational Solutions of First-Order Algebraic Ordinary Difference Equations}
\author{
Thieu N. Vo \thanks{Fractional Calculus, Optimization and Algebra Research Group, Faculty of Mathematics and Statistics, 
Ton Duc Thang University, Ho Chi Minh City, Vietnam. Email: vongocthieu@tdtu.edu.vn}, 
Yi Zhang \thanks{Department of Mathematical Sciences, The University of Texas at Dallas (UTD), USA. 
Supported by the UT Dallas Program: P-1-03246. Email: zhangy@amss.ac.cn}}

\begin{document}

\maketitle

\begin{abstract}
We propose an algebraic geometric approach for studying rational solutions of first-order algebraic ordinary difference equations (\AODEs). 
For an autonomous first-order \AODE, we give an upper bound for the degrees of its rational solutions, and thus derive a complete algorithm for 
computing corresponding rational solutions.
\end{abstract}


\section{Introduction}

An algebraic ordinary difference equation (\AODE) is a difference equation of the form
\[
F(x, y(x), y(x + 1), \cdots, y(x + m))=0,
\]
where $F$ is a nonzero polynomial in $y(x), y(x + 1), \cdots, y(x + m)$ with coefficients in the field $\K(x)$ of rational functions over an algebraically closed field $\K$ of characteristic zero, and~$m \in \N$.
We say that an AO{$\Delta$}E is \emph{autonomous}  if the independent variable $x$ does not appear in it explicitly.
For computational purpose, we may choose $\K=\bar{\mathbb{Q}}$, the field of algebraic numbers.
{\AODEs} naturally appear from various problems, such as symbolic summation~\cite{PWZbook1996, KoutschanThesis}, 
factorization of linear difference operators~\cite{BronsteinPetkovsek1996}, 
analysis of time or space complexity of computer programs with recursive calls~\cite{Eekelen2018}. 
Thus, to determine (closed form) solutions of a given {\AODE}  is a fundamental problem in difference algebra and is of general interest.

Constructive approaches for finding symbolic solutions  of linear difference equations and their applications have been extensively investigated. 
There are well-known algorithms for computing polynomial~\cite{AbramovBronstein1995}, rational~\cite{Abramov1989, Abramov1995, Abramov1998, vanHoeij1998}  and 
hypergeometric~\cite{AbramovPaule1998, Koepf1999, Gosper1978, Koepf1995, Paule1995, Paule1995b, Paule1997,  Petkovsek1992, vanHoeij1999} solutions for linear difference equations. Besides, Karr~\cite{Karr1981, Karr1985}, Kauers and Schneider~\cite{Kauers2006} developed algorithms for determining closed form formulas for finite sums and indefinite nested summation. 
The corresponding implementations are available in computer algebra systems {\tt Maple} and {\tt Mathematica}.  
For more general types of solutions of linear difference equations, we refer to~\cite{HendriksSinger1999, Wolfram2000}. 
Last but not least, Bronstein~\cite{Bronstein2000} and Schneider~\cite{Schneider2005} gave algorithms for calculating solutions 
of parameterized linear difference equations within~$\Pi \Sigma$-fields.

Nevertheless, there are only a few steps in attacking the problem of computing symbolic solutions of (nonlinear) \AODE.
Cohn~\cite{Cohn1965} provided a general algebraic framework for investigating structures of {\AODEs} and their solutions.
Elaydi~\cite{Elaydi2004} summarized some useful techniques for transforming certain nonlinear {\AODE} (such as difference equations of general Riccati type) into linear ones. 
In~\cite{FengGao2008}, the authors gave a polynomial time algorithm for finding polynomial solutions of first-order autonomous {\AODEs}
by utilizing parametrization theory of plane algebraic curves. Using symmetric polynomial theory, 
Shkaravska and Eekelen~\cite{Eekelen2014, Eekelen2018} gave a degree bound for polynomial solutions of high-order non-autonomous {\AODEs}
under a sufficient condition.  

We are mainly interested in rational solutions of first-order \AODEs.
In~\cite{FengGao2008}, Feng, Gao and Huang proposed an algorithm for computing a rational solution for a first-order autonomous {\AODE} provided that a bound for the degree of the rational solution is given.
They also pointed out that they could not bound the degrees of rational solutions through the parametrization technique because the difference version of~\cite[Theorem 3.7]{FengGao06} is not always true (see \cite[Example~4.1]{FengGao2008}). 
We overcome this missing part and present an algorithm for computing such a degree bound.
It is seen that if $y(x) \in \mathbb{K}(x)$ is a nonzero rational solution of an autonomous first-order \AODE, then $y(x+c) \in \mathbb{K}(x,c) \setminus \mathbb{K}(x)$, where $c$ is a constant in a difference extension of $\mathbb{K}$, is again a solution of the given difference equation.
In this paper, we consider a more general problem:

\begin{problem}\label{prob:main}
Let $F \in \mathbb{K}[x,y,z]$ be an irreducible polynomial. Determine a solution~$s \in \mathbb{K}(x,c) \setminus \K(x)$, where $c$ is a transcendental constant, for the following difference equation
\begin{equation} \label{EQ:first-orderADE}
F(x,y(x),y(x+1))=0.
\end{equation}
\end{problem}

A solution in $\K(x,c) \setminus \K(x)$ is called a \emph{strong rational general solution} (compare with \cite[Definition~2.2]{FengGao2008}).
We prove (Theorem~\ref{thm:StrongSolution}) that if the difference equation~\eqref{EQ:first-orderADE} admits a strong rational general solution then its corresponding algebraic curve in $\mathbb{A}^2(\overline{\K(x)})$ defined by $F(x,y,z)=0$ is of genus zero. 
Thus, we can take use of parametric representations of rational curves to transform the original difference equation into an associated difference equation which is of simpler form (see Theorem~\ref{thm:StrongSolution} and Proposition~\ref{prop:birational_transformation}).
The latter difference equation has a special form and it  is called a \emph{separable difference equation} (Remark~\ref{REMARK:separable}). 
We prove (Theorem~\ref{thm:OneToOne}) that there is a one-to-one corresponding between the strong rational general solutions of the given difference equation and those of the associated separable difference equation.
Therefore, the problem of determining a strong rational general solution for a first-order {\AODE} is reduced to that of computing a strong rational general solution of the corresponding separable one. 

For an autonomous first-order \AODE, we give a bound for the degrees of rational solutions of its associated separable difference equation.
Thus, we derive a complete algorithm for computing rational solutions of autonomous first-order {\AODEs} (see Algorithm~\ref{ALGO:firstautonomous}).
To derive a degree bound, we first transform the problem of determining a rational solution of an autonomous separable {\AODE} to that of computing a pair of polynomial solutions of an autonomous first-order system (Subsection~\ref{SUBSEC:reduction}).
Secondly, we use tools from resultant theory to eliminate one dependent variable from the system (Subsection~\ref{SUBSEC:elimination}), and then 
obtain a nontrivial autonomous homogeneous second-order {\AODE} for the other dependent variable.
By using the difference analog of the combiratorial approach in \cite{VoZhang18, Vo2018Computation}, 
we finally present a degree bound for polynomial solutions of an autonomous homogeneous second-order {\AODE} (Subsection~\ref{SUBSEC:polysol}), which is the last key for deriving the complete algorithm.

The rest of the paper is organized as follows. Section~\ref{SECT:transform} is devoted to present an algebraic geometric approach to first-order {\AODEs}. We propose a constructive approach in Section~\ref{SECT:ratsol} and a complete algorithm in Section~\ref{SECT:algo} for computing rational solutions of autonomous 
first-order \AODEs.

\section{An algebraic geometric approach to first-order \AODEs} \label{SECT:transform}

In this section, we study first-order {\AODEs} from an algebraic geometric point of view. 
The idea is inherited from \cite{FengGao,FengGao06,ArocaCanoFengGao,Vo2018Deciding}.
Assume that $F \in \K[x, y, z]$ is a nonzero trivariate polynomial.
We associate the difference equation $F(x,y(x),y(x+1))=0$ with the corresponding algebraic curve in the two dimensional affine plane over the field $\overline{\mathbb{K}(x)}$ of algebraic functions defined by $F(x,y,z)=0$.
We prove that if the given difference equation admits a strong rational general solution then the corresponding algebraic curve is of genus zero.
Therefore we may apply algebraic tools from the theory of rational curves.
In particular, we use rational parametric representations of the algebraic curve to transform the original difference equation 
to a simpler one (see Remark~\ref{REMARK:separable}).
A one-to-one correspondence between strong rational general solutions of the original difference equation and those of the new one is established.

We start this section with a formal definition of strong rational general solution.

\begin{definition}[{See \cite[Definition~3.3]{Vo2018Deciding}}]
Let $F \in \K[x, y, z]$ be a nonzero trivariate polynomial. 
A solution $s$ of the algebraic difference equation $F(x,y(x),y(x+1))=0$ is 
called a \emph{strong rational general solution} if $s=s(x,c) \in \mathbb{K}(x,c) \setminus \mathbb{K}(x)$ 
for some constant $c$ which is transcendental over $\mathbb{K}(x)$.
\end{definition}

The following theorem gives a necessary condition for a first-order algebraic difference equation 
having a strong rational general solution. 


\begin{theorem} \label{thm:StrongSolution}
Let $F$ be an irreducible polynomial in $\mathbb{K}[x,y,z] \setminus \mathbb{K}[x,y]$ and 
consider the difference equation $F(x,y(x),y(x+1))=0$.
If the difference equation admits a strong rational general solution, then 
\begin{enumerate}\renewcommand{\theenumi}{(\roman{enumi})}\renewcommand{\labelenumi}{\theenumi}
	\item\label{it:StrongSolution:irred} $F$ is irreducible as a polynomial in $\overline{\K(x)}[y,z]$, and
	\item\label{it:StrongSolution:genus} the algebraic curve in $\A^2 \left(\overline{\K(x)}\right)$ defined by $F(x,y,z)=0$ 
	is of genus zero. 
\end{enumerate}
\end{theorem}

\begin{proof}
\ref{it:StrongSolution:irred} Let $s(x, c)$ be a strong rational solution of the difference equation $F = 0$.
Consider the ring homomorphism:
\[
 \begin{array}{llll}
  \phi: & \overline{\K(x)}[y,z] & \longrightarrow & \overline{\K(x)}(c) \\
        & G(x,y,z)              & \longmapsto     & G(x,s(x,c),s(x+1,c))
 \end{array}
\]
Let $I$ be the kernel of $\phi$. By assumption, we know that $F \in I$. 
Thus, $I$ is nonzero.
To prove that $F$ is irreducible over $\overline{\K(x)}$, 
we show that $I$ is a principal prime ideal, and that $F$ generates $I$.

Since the zero set of $I$ 
contains a parametric class of 
points $(s(x,c),s(x+1,c))$, the ideal~$I$ is neither a maximal ideal nor the whole ring. 
On the other hand, the homomorphism $\phi$ induces an isomorphism between $\overline{\K(x)}[y,z] / I$ 
and a subring of $\overline{\K(x)}(c)$, which is an integral domain. 
Therefore, the ideal $I$ is a prime ideal.  
Since the Krull dimension of $\overline{\K(x)}[y,z]$ is 2, it follows that $I$ is of height 1. 
By \cite[prop.1.12A, p.7]{Hartshorne}, we conclude that $I$ is principal.

Next, using the technique of Gr\" obner basis, 
we construct a generator of $I$ with coefficients in $\K(x)$. 
In order to do that, we rewrite $I$ as the following form:
$$I=\left \{ H \in \overline{\K(x)}[y,z] \,|\, H(x,s(x,c),s(x+1,c))=0 \right \}.$$
Let $s(x,c)=\frac{P(x,c)}{Q(x,c)}$ and $s(x+1,c)=\frac{R(x,c)}{S(x,c)}$ 
where $P,Q,R,S \in \K[x,c]$ such that $\gcd(P,Q)=\gcd(R,S)=1$. 
By the technique of implicitization \cite[Thm.2, p.138]{Cox2015}, we know that 
$$I=\langle yQ-P,zS-R,1-QSt \rangle \cap \overline{\K(x)}[y,z],$$
where the first component $J$ in the right hand side is an ideal in $\overline{\K(x)}[c,t,y,z]$ 
generated by polynomials $yQ-P,zS-R$ and $1-QSt$. 
Let $\succ$ be the lexicographic order for monomials of $\overline{\K(x)}[c,t,y,z]$ such that $c \succ t \succ y \succ z$. 
Using Buchberger's algorithm, one can determine a reduced Gr\" obner basis $\mathbf{G}$ for $J$ with respect to $\succ$. 
Then $\mathbf{G}$ only contains polynomials in $c,t,y,z$ with coefficients in $\K(x)$. 
After discarding all polynomials involving $c,t$ from $\mathbf{G}$, 
we obtain a reduced Gr\" obner basis $\tilde{\mathbf{G}}$ for $I$ which contains only polynomials in $\K(x)[y,z]$. 
Since $I$ is principal and prime, the basis $\tilde{\mathbf{G}}$ contains only one element, say $G \in \K(x)[y,z]$, 
and $G$ is irreducible over $\overline{\K(x)}$. 

Recall that $F$ is irreducible over $\K(x)$ and $F \in I=\left<G\right>$. 
This implies that $F$ differs from $G$ by a multiplication of a nonzero element in $\K(x)$. 
Hence, $F$ is also irreducible over $\overline{\K(x)}$. 

\ref{it:StrongSolution:genus} As a consequence, the algebraic equation $F(x,y,z)=0$ defines an irreducible algebraic curve 
in the affine plane $\mathbb{A}^2(\overline{\K(x)})$. 
Moreover, this curve can be parametrized by the pair of rational functions $(s(x,c),s(x+1,c))$. 
Hence, by \cite[Theorem 4.7, p.93]{Sendra2008}, the curve is rational. 
We conclude from~\cite[Theorem 4.11, p.95]{Sendra2008} that its genus is zero. 
\end{proof}




The above theorem motivates the following concept.

\begin{definition} \label{DEF:algcurve}
Let $F$ be a nonzero polynomial in $\mathbb{K}[x,y,z]$.
The algebraic curve $\mathcal{C}_F \subset \A^2(\overline{\K(x)})$ defined by $F(x,y,z)=0$
is called the \emph{corresponding algebraic curve} of the first-order algebraic difference equation $F(x,y(x),y(x+1))=0$.
\end{definition}

Due to Theorem~\ref{thm:StrongSolution}, if a first-order {\AODE} admits a strong rational general solution, 
then its corresponding algebraic curve must be of genus zero. 
Therefore, we may apply algebraic tools from parametrization theory of rational curves.

Given a field $\mathbb{L}$ of characteristics zero and a non-constant polynomial $G \in \mathbb{L}[y,z]$, 
the algebraic equation $G(y,z)=0$ implicitly defines an algebraic curve, say $\mathcal{C}_G$, 
in the affine plane $\mathbb{A}^2(\overline{\mathbb{L}})$ over the algebraic closure $\overline{\mathbb{L}}$ of $\mathbb{L}$.
It is well-known that $\mathcal{C}_G$ is of genus zero if and only if there exists a birational transformation
$$
\begin{array}{cccl}
\mathcal{P}: &\mathbb{A}^1 (\overline{\mathbb{L}})& \to & \mathcal{C}_G \subset \mathbb{A}^2 (\overline{\mathbb{L}})\\
& t & \mapsto & (p_1(t),p_2(t))
\end{array}
$$
for some rational functions $p_1,p_2 \in \overline{\mathbb{L}}(t)$.
To construct such a birational transformation is one of the most important problems in parametrization theory of algebraic curves.
For details about parametrization theory of algebraic curves, we refer to \cite{Sendra2008}. 
In particular, there is an algorithm for determining a birational transformation from the affine line to a genus zero algebraic curve over 
the field $\mathbb{L}=\mathbb{K}(x)$ (see \cite{Hoeij2006Solving,Vo2018Deciding}).

The following proposition is a direct consequence of \cite[Proposition~4.3]{Vo2018Deciding}.

\begin{proposition}\label{prop:birational_transformation}
Let $F \in \mathbb{K}(x)[y,z]$ be a non-constant polynomial such that 
the corresponding algebraic curve $\mathcal{C}_F \subset \mathbb{A}^2 (\overline{\mathbb{K}(x)})$ defined by $F(x,y,z)=0$ is of genus zero.
Then there exists a birational transformation $\mathcal{P}:\mathbb{A}^1 (\overline{\mathbb{K}(x)}) \to \mathcal{C}_F$ defined by $\mathcal{P}(x,t)=(p_1(x,t),p_2(x,t))$ for some rational functions $p_1(x,t),p_2(x,t) \in \mathbb{K}(x,t)$.
\end{proposition}

\begin{proof}
Choose the birational transformation $\mathcal{P}$ to be an optimal 
parametrization of $\mathcal{C}_F$ (see \cite[Section~4]{Vo2018Deciding}).
\end{proof}



Note that there exists an algorithm~\cite[Algorithm~1]{Vo2018Deciding} for determining a birational transformation 
(or more precisely, an optimal parametrization) for $\mathcal{C}_F$ in the above proposition. 
Assume that $r \in \K[x]$. The degree of $r$  is  defined to be the maximum of the degree of its numerator and that of its denominator with respect to $x$, 
and we denote it by $\deg_x(r)$ or simply $\deg(r)$.

\begin{theorem}\label{thm:OneToOne}
Let $F(x,y(x),y(x+1))=0$ be an {\AODE} such that its corresponding curve $\mathcal{C}_F$ is of genus zero. 
Assume that $\mathcal{P}(x,t) = (p_1(x,t),p_2(x,t)) \in \K(x,t)^2$ is a birational transformation from 
the affine line $\mathbb{A}^1(\overline{\mathbb{K}(x)})$ to $\mathcal{C}_F$.
Consider the following difference equation
\begin{equation} \label{equ:AssociatedEquation}
p_1(x+1, \omega(x+1)) = p_2 (x,\omega(x)).
\end{equation}
\begin{enumerate}
	\item If $s(x,c) \in \mathbb{K}(x,c) \setminus \mathbb{K}(x)$ is a strong rational general solution of 
	the given difference equation $F(x,y(x),y(x+1))=0$, 
	then there exists a strong rational general solution $\omega(x,c) \in \mathbb{K}(x,c) \setminus \mathbb{K}(x)$ 
	of~\eqref{equ:AssociatedEquation} such that $s(x,c)=p_1(x,\omega(x,c))$.
	\item Conversely, if $\omega(x,c) \in \mathbb{K}(x,c) \setminus \mathbb{K}(x)$ is a strong rational general solution 
	of~\eqref{equ:AssociatedEquation}, then 
	\[s(x,c) = p_1(x,\omega(x,c)) \in \mathbb{K}(x,c) \setminus \mathbb{K}(x)\]
	is a strong rational general solution of $F(x, y(x), y(x + 1)) = 0$.
\end{enumerate}
\end{theorem}

\begin{proof}
\begin{enumerate}
	\item Assume that $s(x,c) \in \mathbb{K}(x,c) \setminus \mathbb{K}(x)$ is 
	a strong rational general solution of the given difference equation, i.e., 
	\[F(x,s(x,c),s(x+1,c))=0.\]
	Then $\mathcal{Q}(x,t)=(s(x,t),s(x+1,t))$ is a parametric representation for $\mathcal{C}_F$.
	By \cite[Lemma 4.17, p. 97]{Sendra2008}, if follows that 
	there exists a function $\omega(x,t) \in \overline{\mathbb{K}(x)}(t)$ such that $\mathcal{Q}(x,t) = \mathcal{P}(x,\omega(x,t))$. 
	In particular, we have
	\begin{align*}
	&s(x,t)=p_1(x,\omega(x,t)), \quad \text{and} \quad s(x+1,t)=p_2(x,\omega(x,t)).
	\end{align*}
	Therefore, we have $p_1(x+1,\omega(x+1,t)) = p_2(x,\omega(x,t))$. 
	Thus, it follows that~$\omega(x,c)$ is a solution of the difference equation \eqref{equ:AssociatedEquation}, 
	and $s(x,c)=p_1(x,\omega(x,c))$.
	Due to the proof of \cite[Lemma 4.17, p.19]{Sendra2008}, we can choose $\omega(x,t)=\mathcal{P}^{-1} \circ \mathcal{Q}(x,t)$.
	Hence, we conclude that $\omega(x,c) \in \mathbb{K}(x,c)$. 
	Since $\deg_c(s)> 0$ and $s(x,c)=p_1(x,\omega(x,c))$, it follows that $\omega(x, c) \in \mathbb{K}(x, c) \setminus \mathbb{K}(x)$. 
	
	\item Assume that $\omega(x,c) \in \mathbb{K}(x,c) \setminus \mathbb{K}(x)$ is a solution of the difference equation \eqref{equ:AssociatedEquation}, i.e.,
	\[p_1(x+1,\omega(x+1,c)) = p_2(x,\omega(x,c)).\]
	Since $\mathcal{P}$ is a birational transformation of $\mathcal{C}_F$, we have $F(x,p_1(x,t),p_2(x,t))=0$.	
	Substituting $t$ by $\omega(x,c)$, we obtain
	\begin{align*}
	0&=F(x,p_1(x,\omega(x,c)),p_2(x,\omega(x,c)))&\\
	&=F(x,p_1(x,\omega(x,c)),p_1(x+1,\omega(x+1,c))).&
	\end{align*}
	Set $s(x, c) = p_1(x,\omega(x,c))$. 
	It follows from the above equalities that $s(x, c)$ is a rational solution of $F(x,y(x),y(x+1))=0$. 
	Since $\deg_t(p_1) > 0$ and $\deg_c(\omega) > 0$, it follows from~\cite[Propostion 1.2, item 11]{Binder1996} 
	that $\deg_c(s) > 0$. In other word, we have $s(x, c)  \in \mathbb{K}(x,c) \setminus \mathbb{K}(x)$.
\end{enumerate}
\end{proof}

The above theorem motivates the following definition. 

\begin{definition} \label{DEF:assdiffeqn}
Using notations in Theorem~\ref{thm:OneToOne}, we call equation~\eqref{equ:AssociatedEquation} 
the \emph{associated difference equation} of $F(x,y(x),y(x+1))=0$.
\end{definition}

From experiments, we find that the associated difference equation is usually simpler (see Example~\ref{EX:non-autonomous}) 
than the original one.
In next sections, we will present a degree bound for rational solutions of the associated difference equation of 
an autonomous first-order \AODE, and thus derive a complete algorithm 
for determining rational solutions of the original one.
The existence of an upper bound for rational solutions of non-autonomous first-order {\AODE} is still open.

The following proposition is a generalization of \cite[Lemma~4.2]{FengGao2008}, which refines the shapes 
of both the {\AODE} 
with strong rational general solutions and the associated one.

\begin{proposition} \label{prop:ReducedEquation}
Let $F \in \mathbb{K}[x,y,z] \setminus \mathbb{K}[x,y]$ be an irreducible polynomial. 
If the algebraic difference equation $F(x,y(x),y(x+1))=0$ admits a strong rational general solution, 
then we have that $\deg_y F = \deg_{z} F$.
Furthermore, in this case, the associated difference equation exists and it must be of the form $$p_1(x,\omega(x+1)) = p_2(x,\omega(x)),$$ 
for some rational functions $p_1,p_2 \in \mathbb{K}(x,y)$ such that 
$$\deg_{y}p_1 = \deg_{y}p_2=\deg_{z}F = \deg_y F.$$
\end{proposition}

\begin{proof}
By item (ii) of Theorem \ref{thm:StrongSolution}, 
the corresponding algebraic curve~$\mathcal{C}_F$ is of genus zero.
By Proposition~\ref{prop:birational_transformation}, there exists a birational transformation $\mathcal{P}$ 
from $\mathbb{A}^1(\overline{\mathbb{K}(x)})$ to~$\mathcal{C}_F$ defined by  $\mathcal{P}(x,t)=(p_1(x,t),p_2(x,t))$ for some $p_1(x,t),p_2(x,t) \in \K(x,t)^2$.
Due to \cite[Theorem 4.21, p.98]{Sendra2008}, we have 
\begin{equation}\label{eq:degree_pi}
\deg_t p_1 = \deg_{z} F, \quad \text{ and } \quad \deg_t p_2 = \deg_y F.
\end{equation} 
By Theorem \ref{thm:OneToOne}, we know that the associated difference equation also admits a strong rational general solution, 
say $\omega(x,c) \in \K(x,c) \setminus \K(x)$, \ie,
$$p_1(x+1,\omega(x+1,c))=p_2(x,\omega(x,c)).$$
By~\cite[Proposition 1.2, item 11]{Binder1996}, it follows that the degrees of the rational functions in both sides of the above equation 
with respect to $c$ is equal to 
$$\deg_t p_1 \cdot \deg_c \omega(x+1,c) = \deg_t p_2 \cdot \deg_c \omega(x,c).$$
Since $\deg_c \omega(x+1,c) = \deg_c \omega(x,c) \geq 1$, we have $\deg_t p_1 = \deg_t p_2$. 
Therefore, we conclude from~\eqref{eq:degree_pi} that $\deg_t p_1 = \deg_t p_2=\deg_{z}F = \deg_y F$.
\end{proof}

\begin{remark} \label{REMARK:separable}
As a consequence of Theorem~\ref{thm:OneToOne} and Proposition~\ref{prop:ReducedEquation}, in order to solve Problem~\ref{prob:main}, 
we only need to consider the class of difference equations of the form
\begin{equation} \label{EQ:separate}
P(x,y(x+1))=Q(x,y(x)).
\end{equation}
for some rational functions $P, Q \in \K(x,z)$ such that $\deg_zP=\deg_zQ$, 
and determine their strong rational general solutions.
We call~\eqref{EQ:separate} a \emph{(rational) separable difference equation}.
If furthermore $P$ and $Q$ are in $\mathbb{K}(z)$,  then we call~\eqref{EQ:separate} 
the \emph{autonomous  separable  difference equation}. 
\end{remark}


\begin{example} \label{EX:non-autonomous}
Consider the following non-autonomous first-order \AODE:
\[
\begin{array}{lll} 
 F(x,y(x),y(x + 1)) & = & (y(x)+x^2+2x+1)y(x + 1)^2 - \\
                    &   & (y(x)+2x^2+2x)y(x) y(x + 1) + x^2y(x)^2 = 0. 
\end{array}     
\]
The corresponding algebraic curve in the affine plane $\mathbb{A}^2(\overline{\mathbb{K}(x)})$ is defined by 
$$F(x,y,z)=(y+x^2+2x+1)z^2-(y+2x^2+2x)yz+x^2y=0.$$
This curve is of genus zero and it admits the following optimal parametrization:
$$\mathcal{P}(x,t)=(p_1(x, t), p_2(x, t) ) = \left( \frac{(xt-1)^2}{t},\frac{(xt-1)^2}{t(t+1)} \right).$$
The associated separable difference equation is
\begin{equation*} 
\frac{((x+1) \omega(x + 1)-1)^2}{\omega(x + 1)}=\frac{(x \omega(x) -1)^2}{\omega(x) (\omega(x) +1)}.
\end{equation*}
By a detailed greatest common divisors argument, we can show that $\omega = \frac{1}{x + c}$ is a strong rational general solution 
of the above equation, 
where $c$ is an arbitrary constant.  Therefore, it follows from Theorem~\ref{thm:OneToOne} that $y(x)=p_1(x,\omega(x))=\frac{c^2}{x+c}$ 
is a strong rational general solution 
of $F(x, y(x), y(x + 1)) = 0$.
\end{example}

\section{Rational solutions of autonomous first-order \AODEs} \label{SECT:ratsol}


In this section, we restrict our consideration to 
the class of autonomous first-order {\AODEs} and 
provide a constructive approach for solving the following problem.

\begin{problem}\label{prob:autonomous}
Let $F \in \mathbb{K}[y,z] \setminus \mathbb{K}[y]$ be an irreducible polynomial. Find all rational solutions of the difference equation
$F(y(x),y(x+1))=0$ if there is any.
\end{problem}

We first apply the results from the previous section to reduce an autonomous first-order {\AODE} 
to an autonomous separable difference equation.
Next, we prove (Proposition~\ref{prop:exist_c_K} and Theorem~\ref{THM:exist_c_C}) that the problem of 
finding rational solutions of an autonomous separable difference equation can be transformed to 
that of determining polynomial solutions of a difference system of order one in two dependent variables.
Using resultant theory, one can always eliminate one dependent variable from that difference system, and thus 
obtain a nontrivial (Theorem~\ref{THM:nonzero_intersection} and Corollary~\ref{COR:elimination}) 
autonomous homogeneous second-order {\AODE} with respect to the other dependent variable. 
Finally, we give a degree bound (Proposition~\ref{PROP:polysoldegbound}) for polynomial solutions of that second-order difference equation.

\subsection{Reduce to an autonomous separable difference equation} \label{SUBSECT:reduceseparable}

Let us consider the following autonomous first-order \AODE:
\begin{equation}\label{eq:autonomous}
F(y(x),y(x+1))=0,
\end{equation}
where $F \in \mathbb{K}[y,z] \setminus \mathbb{K}[y]$ is an irreducible polynomial.
Assume that $y(x) \in \mathbb{K}(x)$ is a rational solution.
First, we observe that every constant function solutions of \eqref{eq:autonomous} are solutions of the algebraic equation $F(x,x)=0$.
Therefore, to avoid triviality, we can assume that the degree of $y(x)$
is at least one.
In this case, the function $y(x+c)$ is a strong rational general solution, where $c$ is a transcendental element over $\mathbb{K}(x)$.
Thus, we can conclude that the problem of finding non-constant rational solutions of an autonomous first-order {\AODE}
is equivalent to that of finding its strong rational general solutions.

Now let us assume that $y(x)$ is a non-constant rational solution of the given difference equation.
Then the difference equation admits a strong rational general solution.
By Theorem~\ref{thm:StrongSolution}, the corresponding 
algebraic curve $\mathcal{C}_F \subset \mathbb{A}^2(\overline{\mathbb{K}(x)})$ defined by $F(y,z)=0$ is of genus zero.
Since the coefficients of $F$ do not involve $x$, there exists a birational transformation 
defined by $\mathcal{P}(t)=(P(t),Q(t)) \in \mathbb{K}(t) \times \mathbb{K}(t)$ from 
the affine line $\mathbb{A}^1(\overline{\mathbb{K}})$ to $\mathcal{C}_F$.
This birational transformation can be chosen to be an optimal parametrization of the algebraic curve in $\mathbb{A}^{2}(\mathbb{K})$ defined by $F(y,z)=0$ (see \cite{Sendra2008}).
Furthermore, by Proposition~\ref{prop:ReducedEquation}, we have $\deg P=\deg Q = \deg_yF$.
By Theorem~\ref{thm:OneToOne}, we conclude that 
finding rational solutions of the associated autonomous separable difference equation $P(y(x+1))=Q(y(x))$ 
of~\eqref{eq:autonomous} is enough for solving Problem~\ref{prob:autonomous}.


\subsection{Reduce to the problem of finding polynomial solutions of an autonomous first-order difference system} \label{SUBSEC:reduction}

Based on arguments of the previous subsection, 
we now restrict our consideration further to the class of autonomous separable difference equations.


\begin{problem}\label{prob:separable_rational}
Let $P_1,P_2,Q_1,Q_2$ be polynomials in $\mathbb{K}[z] \setminus \{ 0 \}$ such 
that $\gcd(P_1,Q_1)=\gcd(P_2,Q_2)=1$ and $\deg \frac{P_1}{Q_1} = \deg \frac{P_2}{Q_2} = n \geq 1$.
Find all rational solutions of the difference equation
\begin{equation}\label{eq:separable_rational}
\frac{P_1(y(x+1))}{Q_1(y(x+1))} = \frac{P_2(y(x))}{Q_2(y(x))}.
\end{equation}
\end{problem}

If $n=1$, then the difference equation \eqref{eq:separable_rational} is of general Riccati type (see \cite[Section~2.6]{Elaydi2004}).
A difference equation of general Riccati type can be solved by transforming it into a second-order linear difference equation 
with polynomial coefficients and then solving the latter one.
Details about the transformation can be found in \cite[Section~2.6]{Elaydi2004}.
Unfortunately, this method can not be generalized immediately to the arbitrary degree case.

We give a new method that works for the general case. 
The first step of this approach is to reduce Problem~\ref{prob:separable_rational}  
to that of determining polynomial solutions of a system of autonomous first-order {\AODEs} 
(see Problem~\ref{prob:system_polynomial}). 
For doing that, we first need some technical lemmas.

\begin{lemma} \label{lem:Divide}
Let $P,Q,R,S \in \mathbb{K}[x]$ be nonzero polynomials such that $\gcd(P,Q)=\gcd(R,S)=1$ and $\frac{P}{Q}=\frac{R}{S}$. 
Then there exists a constant $c$ in $\mathbb{K}$ such that $P=cR$ and $Q=cS$.
\end{lemma}

\begin{proof}
Straightforward.
\end{proof}


\begin{lemma}\label{lem:RS_coprime}
Let $R,S \in \mathbb{K}[z,w]$ be homogeneous polynomials such that $\gcd (R,S)=1$, and $A, B$ be polynomials in $\K[x]$ 
such that $\gcd (A,B)=1$. Then $$\gcd \left( R(A(x),B(x)),S(A(x),B(x)) \right) = 1.$$
\end{lemma}

\begin{proof}
To avoid triviality, we may assume that $R$ and $S$ are of degrees at least~$1$.
We first prove the lemma for the case $\deg R = \deg S = 1$.
In this case, without loss of generality, we can assume that $R=z+aw$ and $S=bz+cw$ for some $a,b,c \in \mathbb{K}$ such that $(b,c) \neq (0,0)$.
Since $\gcd(R,S)=1$, we have $c-ab \neq 0$.
Then we have
\begin{align*}
\gcd(R(A,B),S(A,B)) & = \gcd (R(A,B),S(A,B) - bR(A,B))&\\
& = \gcd(A+aB,(c-ab)B)&\\
& = \gcd(A+aB,B) = \gcd(A,B) =1.&
\end{align*} 

Next, we assume that $R$ and $S$ have positive degrees.
Let $\deg R=m,\, \deg S=n$. 
Since $\mathbb{K}$ is algebraically closed, the homogeneous polynomials $R$ and $S$ can be factored into the products 
of linear homogeneous polynomials, say
\[
 R= R_1^{r_1} \cdot \ldots \cdot R_m^{r_m} , \text{ and } \ \  S= S_1^{s_1} \cdot \ldots \cdot S_n^{s_n},
\]
where $R_i,S_j$ are linear homogeneous polynomials in $\K[z, w]$ and $r_i,s_j$ are positive integers.
Since $R$ and $S$ are coprime, $R_i$ and $S_j$ are also coprime for every $i,\,j$. Using the above argument, we have that
$\gcd(R_i(A,B),S_j(A,B)) = 1$
for every $i,j$.
Hence, we conclude that $\gcd(R(A,B),S(A,B)) =1$.
\end{proof}

\begin{proposition}\label{prop:exist_c_K}
Using notations in Problem~\ref{prob:separable_rational}, we set
\begin{equation*}
\tilde{P}_i(z,w) = w^nP_i \left(\frac{z}{w}\right), \quad \text{ and } \quad \tilde{Q}_i(z,w) = w^nQ_i \left(\frac{z}{w}\right),
\end{equation*}
which are homogeneous polynomials of degree $n$ in $\mathbb{K}[z,w]$, $i = 1, 2$.
Assume that the function $\frac{A(x)}{B(x)}$ is a solution of equation~\eqref{eq:separable_rational}, 
where $A,B \in \mathbb{K}[x]$ and~$\gcd (A,B)=1$. 
Then there exists a constant $c \in \mathbb{K}$ such that 
\begin{equation}\label{sys:exist_c_C}
\left \{ 
\begin{aligned}
\tilde{P}_1(A(x+1),B(x+1)) = c \cdot \tilde{P}_2(A(x),B(x)),\\
\tilde{Q}_1(A(x+1),B(x+1)) = c \cdot \tilde{Q}_2(A(x),B(x)).
\end{aligned}
\right.
\end{equation}
\end{proposition}

\begin{proof}
Substituting $\frac{A(x)}{B(x)}$ into~\eqref{eq:separable_rational} and clearing numerators and denominators of the both sides, 
we obtain 
\begin{equation} \label{eq:separable_AB}
\frac{\tilde{P}_1(A(x+1),B(x+1))}{\tilde{Q}_1(A(x+1),B(x+1))} = 
\frac{\tilde{P}_2(A(x),B(x))}{\tilde{Q}_2(A(x),B(x))}.
\end{equation}

Next, we prove that the numerators and the denominators of the right hand side for the above equation are coprime.
By Lemma~\ref{lem:RS_coprime}, it suffices to prove that $\tilde{P}_2$ and $\tilde{Q}_2$ are coprime polynomials in $\mathbb{K}[z,w]$.
Since $\deg \frac{P_2}{Q_2}=n$, without loss of generality, we can assume that $\deg P_2(z) = m$, $\deg Q_2(z) = n$ and $m \leq n$.
In this case, we have
\begin{equation} \label{EQ:homogenization}
 \tilde{P}_2(z,w) = w^{n-m} \cdot \left[w^m P_2\left( \frac{z}{w} \right) \right], \text{ and } \ \ 
 \tilde{Q}_2(z,w) = w^n Q_2\left(\frac{z}{w}\right).
\end{equation}
Note that $w^m P_2\left( \frac{z}{w} \right)$ and $w^n Q_2\left(\frac{z}{w}\right)$ are homogenizations of $P_2(z)$ and $Q_2(z)$, respectively. 
Since $P_2(z)$ and $Q_2(z)$ are coprime, we have $w^m P_2\left( \frac{z}{w} \right)$ and $\tilde{Q}_2(z,w)$ are coprime.
Furthermore, on account of $\deg Q_2(z) = n$, we see that $\tilde{Q}_2(z,w) = az^n + w Q(z,w)$ for some $a \in \mathbb{K} \setminus \{0\}$ and $Q(z,w) \in \mathbb{K}[z,w]$.
This implies that $w^{n-m}$ and $\tilde{Q}_2(z,w)$ are coprime, too.
Hence, we conclude from~\eqref{EQ:homogenization} that $\tilde{P}_2(z,w)$ and $\tilde{Q}_2(z,w)$ are coprime.

Similarly, we can also show that the numerator and the denominator of the left hand side of~\eqref{eq:separable_AB} are coprime. 
Therefore, it follows from Lemma~\ref{lem:Divide} that the claim of this proposition holds.
\end{proof}

It is straightforward to see that any pair of solutions $(A(x), B(x))$ of~\eqref{sys:exist_c_C} gives rise to a rational solution 
of~\eqref{eq:separable_rational}. Therefore, it suffices to consider polynomial solutions of~\eqref{sys:exist_c_C}. 
However, there is a new unknown constant appeared in~\eqref{sys:exist_c_C}. 
In order to determine the exact values for the constant $c$ in the above proposition, we introduce 
the following definition.

\begin{definition}\label{def:C}
Using notations in Problem~\ref{prob:separable_rational}, we rewrite $P_1, Q_1, P_2, Q_2 \in \mathbb{K}[z]$ as follows:
\begin{align*}
P_1(z)=r_{1,\textbf{high}} \cdot z^{k_{1,\textbf{high}}} + \cdots + r_{1,\textbf{low}} \cdot z^{k_{1,\textbf{low}}},\nonumber\\
Q_1(z)=s_{1,\textbf{high}} \cdot z^{\ell_{1,\textbf{high}}} + \cdots + s_{1,\textbf{low}} \cdot z^{\ell_{1,\textbf{low}}},\nonumber\\
P_2(z)=r_{2,\textbf{high}} \cdot z^{k_{2,\textbf{high}}} + \cdots + r_{2,\textbf{low}} \cdot z^{k_{2,\textbf{low}}},\nonumber\\
Q_2(z)=s_{2,\textbf{high}} \cdot z^{\ell_{2,\textbf{high}}} + \cdots + s_{2,\textbf{low}} \cdot z^{\ell_{2,\textbf{low}}},
\end{align*}
where $k_{i, \textbf{high}} > k_{i,\textbf{low}}$, and $\ell_{i, \textbf{high}} > \ell_{i,\textbf{low}}$, $i = 1, 2$.
Define a set $\mathcal{C} \subseteq \mathbb{K}$ associated to~\eqref{eq:separable_rational} recursively as follows:

\begin{itemize}
\item[(1)] Set $\mathcal{C} = \left \{ c \in \mathbb{K} \,|\, \exists \alpha \in \mathbb{K}\,:\,
P_1(\alpha)=cP_2(\alpha) \textbf{ and }
Q_1(\alpha)=cQ_2(\alpha) \right \}.$
\item[(2)] If $k_{1,\textbf{high}} = k_{2,\textbf{high}}$ then add $\frac{r_{1,\textbf{high}}}{r_{2,\textbf{high}}}$ and $0$ to $\mathcal{C}$,
\item[(3)] If $l_{1,\textbf{high}} = l_{2,\textbf{high}}$ then add $\frac{s_{1,\textbf{high}}}{s_{2,\textbf{high}}}$ and $0$ to $\mathcal{C}$,
\item[(4)] If $k_{1,\textbf{low}} = k_{2,\textbf{low}}$ then add $\frac{r_{1,\textbf{low}}}{r_{2,\textbf{low}}}$ and $0$ to $\mathcal{C}$,
\item[(5)] If $l_{1,\textbf{low}} = l_{2,\textbf{low}}$ then add $\frac{s_{1,\textbf{low}}}{s_{2,\textbf{low}}}$ and $0$ to $\mathcal{C}$.
\end{itemize}
We call $\mathcal{C}$ \emph{the set of constant candidates} of~\eqref{eq:separable_rational}.
\end{definition}


\begin{theorem}\label{THM:exist_c_C}
The constant $c$ in Proposition~\ref{prop:exist_c_K} is exactly one element 
in the set of constant candidates of~\eqref{eq:separable_rational}. 
\end{theorem}
\begin{proof}
We keep the notations of Proposition~\ref{prop:exist_c_K} and Definition~\ref{def:C}. 
Let $\frac{A(x)}{B(x)}$ be a solution of~\eqref{eq:separable_rational}, 
where $A(x)$ and $B(x)$ are the same as that in Proposition~\ref{prop:exist_c_K}.
Without loss of generality, we may further assume that 
the leading coefficient of $A(x)$ is $\alpha \in \mathbb{K} \setminus \{0\}$ and that of $B(x)$ is $1$. 
Let $P_1,P_2,Q_1,Q_2$ be polynomials in Definition~\ref{def:C}. 
And $\mathcal{C}$ is the set of constant candidates of~\eqref{eq:separable_rational}.
Consider the following three cases.

\begin{itemize}
\item[] \textbf{Case 1.} $\deg A = \deg B$.
By comparing the leading coefficients of polynomials in system~\eqref{sys:exist_c_C}, we obtain
$$
\left \{
\begin{aligned}
\tilde{P}_1(\alpha,1) = c \cdot \tilde{P}_2(\alpha,1),\\
\tilde{Q}_1(\alpha,1) = c \cdot \tilde{Q}_2(\alpha,1).
\end{aligned}
\right.
$$
In other words, we have
$$
\left \{
\begin{aligned}
P_1(\alpha) = c \cdot P_2(\alpha),\\
Q_1(\alpha) = c \cdot Q_2(\alpha).
\end{aligned}
\right.
$$
Hence, it follows from Definition~\ref{def:C} that $c \in \mathcal{C}$ .

\item[] \textbf{Case 2.} $\deg A > \deg B$.
The following table shows degrees and leading coefficients of polynomials appearing in system~\eqref{sys:exist_c_C}.

\begin{tabular}{ c c c }
\hline
polynomial & degree & leading coefficient\\
\hline
$\tilde{P}_1(A(x+1),B(x+1))$ & $k_{1,\textbf{high}} \deg A + (n-k_{1,\textbf{high}}) \deg B$ & $r_{1,\textbf{high}} \alpha^{k_{1,\textbf{high}}}$\\
$\tilde{Q}_1(A(x+1),B(x+1))$& $\ell_{1,\textbf{high}} \deg A + (n-\ell_{1,\textbf{high}}) \deg B$ & $s_{1,\textbf{high}} \alpha^{\ell_{1,\textbf{high}}}$\\
$\tilde{P_2}(A(x),B(x))$& $k_{2,\textbf{high}} \deg A + (n-k_{2,\textbf{high}}) \deg B$ & $r_{2,\textbf{high}} \alpha^{k_{2,\textbf{high}}}$\\
$\tilde{Q_2}(A(x),B(x))$& $\ell_{2,\textbf{high}} \deg A + (n-\ell_{2,\textbf{high}}) \deg B$ & $s_{2,\textbf{high}} \alpha^{\ell_{2,\textbf{high}}}$\\
\hline
\end{tabular}

Consider the following three subcases.
	\begin{itemize}
	\item[] \textbf{Subcase 2.1.} $\deg \tilde{P}_1(A(x+1),B(x+1)) \leq \deg \tilde{P_2}(A(x),B(x))$.
	
	This implies that $$k_{1,\textbf{high}} \deg A + (n-k_{1,\textbf{high}}) \deg B \leq k_{2,\textbf{high}} \deg A + (n-k_{2,\textbf{high}}) \deg B.$$ 
	Or equivalently, we have $k_{1,\textbf{high}} \leq k_{2,\textbf{high}}$.
	By comparing the coefficients of terms with degree $\deg \tilde{P_2}(A(x),B(x))$ in the first equation of system~\eqref{sys:exist_c_C}, we obtain
	\begin{equation*}
	c \cdot r_{2,\textbf{high}} \alpha^{k_{2,\textbf{high}}} = \left \{
	\begin{aligned}
		r_{1,\textbf{high}} \alpha^{k_{1,\textbf{high}}} & \text{ if } k_{1,\textbf{high}} = k_{2,\textbf{high}},\\
		0 & \text{ if } k_{1,\textbf{high}} < k_{2,\textbf{high}}.
	\end{aligned}
	\right.
	\end{equation*}
	Hence, we have 
	\begin{equation*}
	c = \left \{
	\begin{aligned}
		\frac{r_{1,\textbf{high}}}{r_{2,\textbf{high}}} & \text{ if } k_{1,\textbf{high}} = k_{2,\textbf{high}},\\
		0 & \text{ if } k_{1,\textbf{high}} < k_{2,\textbf{high}},
	\end{aligned}
	\right.
	\end{equation*}
	which belongs to $\mathcal{C}$.

	\item[] \textbf{Subcase 2.2} $\deg \tilde{Q}_1(A(x+1),B(x+1)) \leq \deg \tilde{Q_2}(A(x),B(x))$.\\
	This inequality is equivalent to $l_{1,\textbf{high}} \leq l_{2,\textbf{high}}$.
	Using the similar argument as that in the above subcase, we obtain
	\begin{equation*}
	c = \left \{
	\begin{aligned}
		\frac{s_{1,\textbf{high}}}{s_{2,\textbf{high}}} & \text{ if } l_{1,\textbf{high}} = l_{2,\textbf{high}},\\
		0 & \text{ if } l_{1,\textbf{high}} < l_{2,\textbf{high}},
	\end{aligned}
	\right.
	\end{equation*}
	which belongs to $\mathcal{C}$.

	\item[] \textbf{Subcase 2.3.} None of the above two subcases happen.\\
	In this case, we have $k_{1,\textbf{high}} > k_{2,\textbf{high}}$ and $s_{1,\textbf{high}} > s_{2,\textbf{high}}$.
	Thus, we have
	\begin{align*}
	n &= \deg \frac{P_1}{Q_1} = \max \{k_{1,\textbf{high}}, s_{1,\textbf{high}}\} \\
	&> \max \{k_{2,\textbf{high}}, s_{2,\textbf{high}}\} = \deg \frac{P_2}{Q_2} = n.
	\end{align*}
	This is impossible. Hence, this subcase can not happen.
	\end{itemize}

\item[] \textbf{Case 3.} $\deg A < \deg B$. 
The following table shows degrees and leading coefficients of polynomials appearing in system~\eqref{sys:exist_c_C}.

\begin{tabular}{ c c c }
\hline
polynomial & degree & leading coefficient\\
\hline
$\tilde{P}_1(A(x+1),B(x+1))$ & $k_{1,\textbf{low}} \deg A + (n-k_{1,\textbf{low}}) \deg B$ & $r_{1,\textbf{low}} \alpha^{k_{1,\textbf{low}}}$\\
$\tilde{Q}_1(A(x+1),B(x+1))$& $\ell_{1,\textbf{low}} \deg A + (n-\ell_{1,\textbf{low}}) \deg B$ & $s_{1,\textbf{low}} \alpha^{\ell_{1,\textbf{low}}}$\\
$\tilde{P_2}(A(x),B(x))$& $k_{2,\textbf{low}} \deg A + (n-k_{2,\textbf{low}}) \deg B$ & $r_{2,\textbf{low}} \alpha^{k_{2,\textbf{low}}}$\\
$\tilde{Q_2}(A(x),B(x))$& $\ell_{2,\textbf{low}} \deg A + (n-\ell_{2,\textbf{low}}) \deg B$ & $s_{2,\textbf{low}} \alpha^{\ell_{2,\textbf{low}}}$\\
\hline
\end{tabular}

Consider the following three subcases.

	\begin{itemize}
	\item[] \textbf{Subcase 3.1.} $\deg \tilde{P}_1(A(x+1),B(x+1)) \leq \deg \tilde{P_2}(A(x),B(x))$.
	
	This means that $$k_{1,\textbf{low}} \deg A + (n-k_{1,\textbf{low}}) \deg B \leq k_{2,\textbf{low}} \deg A + (n-k_{2,\textbf{low}}) \deg B.$$ 
	Or equivalently, $k_{1,\textbf{low}} \geq k_{2,\textbf{low}}$.
	Using a similar argument as that in \textbf{Subcase 2.1}, we obtain
	\begin{equation*}
	c = \left \{
	\begin{aligned}
		\frac{r_{1,\textbf{low}}}{r_{2,\textbf{low}}} & \text{ if } k_{1,\textbf{low}} = k_{2,\textbf{low}},\\
		0 & \text{ if } k_{1,\textbf{low}} > k_{2,\textbf{low}},
	\end{aligned}
	\right.
	\end{equation*}
	which belongs to $\mathcal{C}$.

 	\item[] \textbf{Subcase 3.2.} $\deg \tilde{Q}_1(A(x+1),B(x+1)) \leq \deg \tilde{Q_2}(A(x),B(x))$.
	
	This inequality is equivalent to $l_{1,\textbf{low}} \geq l_{2,\textbf{low}}$.
	Using a similar argument as that in \textbf{Subcase 2.1}, we have
	\begin{equation*}
	c = \left \{
	\begin{aligned}
		\frac{s_{1,\textbf{low}}}{s_{2,\textbf{low}}} & \text{ if } s_{1,\textbf{low}} = s_{2,\textbf{low}},\\
		0 & \text{ if } s_{1,\textbf{low}} > s_{2,\textbf{low}},
	\end{aligned}
	\right.
	\end{equation*}
	which also belongs to $\mathcal{C}$.	
	
	\item[] \textbf{Subcase 3.3.} None of the subcases \textbf{3.1} and \textbf{3.2} occur.\\
	In this case, we have $k_{1,\textbf{low}} < k_{2,\textbf{low}}$ and  $l_{1,\textbf{low}} < l_{2,\textbf{low}}$.
	Thus, we have  $k_{2,\textbf{low}} \geq 1$ and $l_{2,\textbf{low}} \geq 1$.
	In other words, both polynomials $P_2(z)$ and $Q_2(z)$ are divisible by $z$.
	This contradicts to the assumption that $\gcd (P_2,Q_2)=1$.
	Hence, this subcase cannot happen.
	\end{itemize}
\end{itemize}
Above all, we conclude that in each case we always have $c \in \mathcal{C}$.
\end{proof}

It might happen that the set of constant candidates $\mathcal{C}$ of~\eqref{eq:separable_rational} is an infinite set.
In this case, there are infinitely many choices for the constant $c$ appearing in the difference system \eqref{sys:exist_c_C}.
We will see in Proposition~\ref{prop:C_infinite} and~\ref{prop:f=f} below that the given difference equation~\eqref{eq:separable_rational} 
in this case is rather simple and it only has constant solutions.

\begin{proposition}\label{prop:C_infinite}
Let $P_i(z), Q_i(z) \in \K[z]$ be polynomials in~\eqref{eq:separable_rational}, $i = 1, 2$. 
Then the set of constant candidates of~\eqref{eq:separable_rational} is an infinite set 
if and only if $\frac{P_1(z)}{Q_1(z)} = \frac{P_2(z)}{Q_2(z)}$. 
\end{proposition}
\begin{proof}
Assume that $\mathcal{C}$ is the set of constant candidates of~\eqref{eq:separable_rational}. 
From Definition~\ref{def:C}, it is clear that $\mathcal{C}$ is an infinite set 
if and only if the set
$$\mathcal{\bar{\mathcal{C}}} = \left \{ c \in \mathbb{K} \,|\, \exists \alpha \in \mathbb{K}\,:\,
P_1(\alpha)=cP_2(\alpha) \textbf{ and }
Q_1(\alpha)=cQ_2(\alpha) \right \}$$
is an infinite set.

Note that if $P_1(\alpha)=cP_2(\alpha)$ and $Q_1(\alpha)=cQ_2(\alpha)$ for some $\alpha \in \K$, then we always have
\begin{equation}\label{eq:det}
\text{det} 
\begin{pmatrix}
P_1(\alpha)&P_2(\alpha)\\
Q_1(\alpha)&Q_2(\alpha)
\end{pmatrix} = 0.
\end{equation}
For each $\alpha \in \mathbb{K}$ satisfying \eqref{eq:det}, since $\gcd(P_2, Q_2) = 1$, 
there exists a unique $c \in \mathbb{K}$ such that $P_1(\alpha)=cP_2(\alpha)$ and $Q_1(\alpha)=cQ_2(\alpha)$.
Therefore, $\bar{\mathcal{C}}$ is an infinite set
if and only if the algebraic equation 
$$\text{det} 
\begin{pmatrix}
P_1(z)&P_2(z)\\
Q_1(z)&Q_2(z)
\end{pmatrix} = 0$$
has infinitely many solutions. 
The latter happens if and only if the left hand side is actually the zero polynomial. 
In other words, we have $\frac{P_1(z)}{Q_1(z)} = \frac{P_2(z)}{Q_2(z)}$.
\end{proof}

In~\eqref{eq:separable_rational}, 
if $\frac{P_1(z)}{Q_1(z)} \neq \frac{P_2(z)}{Q_2(z)}$, then we can construct $\mathcal{\bar{C}}$ of the above proof 
in an algorithmic way. There, we can compute the set of constant candidates of~\eqref{eq:separable_rational} in this case. 

The following proposition provides an answer for Problem~\ref{prob:separable_rational} 
if $\frac{P_1(z)}{Q_1(z)} = \frac{P_2(z)}{Q_2(z)}$.

\begin{proposition}\label{prop:f=f}
Let $f(z) \in \mathbb{K}(z)$ be a non-constant rational function.
Then rational solutions of the difference equation $f(y(x+1)) = f(y(x))$ are only constants in $\K$.
\end{proposition}

\begin{proof}
Assume that $s(x)$ is a rational solution of the given difference equation. 
Then we have $f(s(x+1))=f(s(x))$ for every $x \in \mathbb{K}$ except finitely many points in $\K$.
Let us fix a constant $x_0 \in \mathbb{K}$ such that $f(s(x_0))$ is defined, 
and define the new rational function~$g$ by $g(x)=f(s(x))-f(s(x_0))$.
It is clear that $g(n+x_0)=0$ for every $n \in \mathbb{Z}$ except finitely many points in $\Z$.
This only happens when $g$ is the zero function.
Therefore, there exists a constant $c \in \mathbb{K}$ such that $f(s(x))-c$ is the zero function.
This means that the function $s(x)$ is a zero of the nonzero rational function $f(z)-c \in \mathbb{K}(z)$.
Thus, the rational function $s(x)$ is algebraic over $\mathbb{K}$.
Since $\mathbb{K}$ is algebraically closed, we conclude that $s(x) \in \mathbb{K}$.
\end{proof}

Using the above proposition, 
we obtain a simple answer for Problem~\ref{prob:separable_rational} 
if $\frac{P_1}{Q_1} = \frac{P_2}{Q_2}$.
In fact, in this case, every rational solution of the given autonomous separable difference equation is a constant function.

Now we can assume that $\frac{P_1}{Q_1} \neq \frac{P_2}{Q_2}$.
By Theorem~\ref{THM:exist_c_C}, the set of constant candidates~$\mathcal{C}$ of~\eqref{eq:separable_rational} 
is always a finite set.
As we have seen in Proposition~\ref{prop:exist_c_K}, 
in order to find a rational solution $\frac{A(x)}{B(x)}$ 
of the given separable difference equation~\eqref{eq:separable_rational}, 
we need to determine a pair of polynomials $(A(x),B(x))$ 
satisfying the difference system~\eqref{sys:exist_c_C} for each $c$ in the finite set~$\mathcal{C}$.
Therefore, we reduce Problem~\ref{prob:separable_rational} to that of finding polynomial solutions 
of finitely many difference systems of the form~\eqref{sys:exist_c_C}.
The latter one is addressed in next subsection.

\subsection{Reduce to the problem of finding polynomial solutions of an autonomous second-order \AODE} \label{SUBSEC:elimination}

In this subsection, we consider the problem of finding polynomial solutions of a difference system of the form~\eqref{sys:exist_c_C}. 
Let $\tilde{P}_2, \tilde{Q}_2 \in \K[z]$ be polynomials and $c$ be the constant in~\eqref{sys:exist_c_C}. 
By Theorem~\ref{THM:exist_c_C}, we can compute all possible values of $c$ by calculating the set of constant candidates 
of~\eqref{eq:separable_rational}. Assume that $c$ is given in~\eqref{sys:exist_c_C}. 
Note that if $c = 0$, then it follows from the fact that $\K$ is algebraically closed 
that system~\eqref{sys:exist_c_C} only has constant solutions in this case. 
Therefore, we may further assume that~$c \neq 0$.
Then we can replace $c \tilde{P}_2$ by $\bar{P}_2$ and $c \tilde{Q}_2$ by $\bar{Q}_2$. 
Thus, without loss of generality, we can also assume that $c=1$.
To be more specific, we focus on the following question.

\begin{problem}\label{prob:system_polynomial}
Let $P_1,P_2,Q_1,Q_2 \in \K[z] \setminus \{ 0 \}$ be polynomials such that~$\gcd(P_1,Q_1)=\gcd(P_2,Q_2)=1$ 
and $\deg \frac{P_1}{Q_1} = \deg \frac{P_2}{Q_2} = n \geq 1$.
Set
\begin{equation*}
\tilde{P}_i(z,w) = w^nP_i \left(\frac{z}{w}\right), \quad \text{ and } \quad \tilde{Q}_i(z,w) = w^nQ_i \left(\frac{z}{w}\right),
\end{equation*}
which are homogeneous polynomials of degree $n$ in $\mathbb{K}[z,w]$, $i = 1, 2$.
Determine all polynomials $A(x),B(x) \in \mathbb{K}[x]$ with $\gcd(A,B)=1$ such that 
\begin{equation}\label{sys:system_polynomial}
\left \{ 
\begin{aligned}
\tilde{P}_1(A(x+1),B(x+1)) = \tilde{P}_2(A(x),B(x)),\\
\tilde{Q}_1(A(x+1),B(x+1)) = \tilde{Q}_2(A(x),B(x)).
\end{aligned}
\right.
\end{equation}
\end{problem}

Next, we will present an algorithm for deriving nonzero second-order 
{\AODEs} for $A(x)$ and $B(x)$ from system~\eqref{sys:system_polynomial}, respectively.
The technique we use is similar to the prolongation-relaxation approach used in differential algebra (see \cite[Section~87-88]{Ritt1932}).

\begin{algo}\label{alg:system_to_single_DEs}
Given difference system \eqref{sys:system_polynomial}, compute nonzero second-order 
{\AODEs} for $A(x)$ and $B(x)$, respectively.
\begin{itemize}
\item[(1)] Let $I \subseteq \mathbb{K}[w_0,w_1,w_2,z_0,z_1,z_2]$ be the ideal generated by the following polynomials
	\begin{align*}
	&\tilde{P}_1(z_1,w_1) - \tilde{P}_2(z_0,w_0), \ \ \tilde{Q}_1(z_1,w_1) - \tilde{Q}_2(z_0,w_0),\\
	&\tilde{P}_1(z_2,w_2) - \tilde{P}_2(z_1,w_1), \ \ \tilde{Q}_1(z_2,w_2) - \tilde{Q}_2(z_1,w_1).
	\end{align*}
Using Gr\"{o}bner bases or resultants, compute nonzero elements $F_A \in I \cap \K[z_0, z_1, z_2]$ 
and $F_B \in I \cap \K[w_0, w_1, w_2]$. 
\item[(2)] Return $F_A(A(x), A(x+ 1), A(x + 1)) = 0$ and $F_B(B(x), B(x + 1), B(x + 2)) = 0$.
\end{itemize}
\end{algo}

The termination of the above algorithm is clear. 
The aim of this subsection is to prove the correctness of the above algorithm. 
To be more specific, we will prove that one can always derive nonzero second-order 
{\AODEs} for $A(x)$ and~$B(x)$ from system~\eqref{sys:system_polynomial}, respectively. 

In our arguments, we use polynomial resultants to derive those nonzero second-order 
{\AODEs} for $A(x)$ and $B(x)$, respectively.
For polynomials $P(z)$ and $Q(z)$ in $\K[z]$, we denote by $\Res_{z}(P(z),Q(z))$ the resultant of $P$ and $Q$ with respect to $z$. 
The resultant is a polynomial over $\mathbb{Z}$ in the coefficients of $P$ and $Q$. 
For details about resultant theory, we refer to \cite[Chapter 3, \S 6]{Cox2015}. 

\begin{proposition}\label{prop:resultant_commute_substitution}
Let $f(z,w)$ and $g(z,w)$ be nonzero polynomials in $\mathbb{K}[z,w]$. 
Let $R(z) = \Res_{w}(f,g) \in \mathbb{K}[z]$ be the resultant of $f$ and $g$ with respect to $w$.
Then for any $c \in \mathbb{K}$ we have
$$R(c) = \Res_{w}(f(c,w),g(c,w)).$$
In other words, the resultant computation always commutes with the substitution of polynomials.
\end{proposition}

\begin{proof}
It is clear from the definition of resultants (see~\cite[Definition~2, p.162]{Cox2015}).
\end{proof}

\begin{lemma}\label{lem:2resultant}
Let $P_1,P_2,Q_1,Q_2$ be nonzero polynomials in $\mathbb{K}[z]$ with $\gcd(P_1,Q_1)=\gcd(P_2,Q_2)=1$. Set
\[R = \Res_{z_2} \left( P_1(z_1) - wP_2(z_2),Q_1(z_1) - wQ_2(z_2) \right).\]
Then $R$ is a nonzero polynomial in $\mathbb{K}[w,z_1]$.
\end{lemma}

\begin{proof}
Suppose that $R=0$.
Due to \cite[Proposition~3, p.~163]{Cox2015}, there exists $z_{2,0} \in \overline{\mathbb{K}(w,z_1)}$ such that 
\begin{equation}\label{eq:PQ}
P_1(z_1) = wP_2(z_{2,0}), \quad \text{and} \quad
Q_1(z_1) = wQ_2(z_{2,0}).
\end{equation}
Since $\gcd(P_2,Q_2)=1$, either $P_2(z_{2,0})$ or $Q_2(z_{2,0})$ is nonzero.
Without loss of generality, we can assume that $Q_2(z_{2,0}) \neq 0$.
From \eqref{eq:PQ}, we obtain $\frac{P_2(z_{2,0})}{Q_2(z_{2,0})} = \frac{P_1(z_1)}{Q_1(z_1)}$. 
This implies that $z_{2,0}$ is a zero of the nonzero rational function $\frac{P_2(z_2)}{Q_2(z_2)} - \frac{P_1(z_1)}{Q_1(z_1)}$ in $\mathbb{K}(z_1, z_2)$.
Therefore, we have $z_{2,0} \in \overline{\mathbb{K}(z_1)}$. 
It follows that $Q_2(z_{2,0})  \in \overline{\mathbb{K}(z_1)}$. 
However, since $w$ is transcendental over $\overline{\mathbb{K}(z_1)}$,
the only possibility for $Q_1(z_1) = wQ_2(z_{2,0})$ to be true is that $Q_2(z_{2,0})=0$, a contradiction. 
\end{proof}

\begin{lemma}\label{lem:4resultant}
Let $P_i, Q_i$ be nonzero polynomials in $\mathbb{K}[z]$ with $\gcd(P_i,Q_i)= 1$, $i = 1, 2$. 
Set $J$ to be the ideal of $\mathbb{K}[w_0,w_1,w_2,z_0,z_1,z_2]$ generated by the following polynomials:
\begin{align*}
w_1P_1(z_1) - w_0P_2(z_0), \ \ w_1Q_1(z_1) - w_0Q_2(z_0),\\
w_2P_1(z_2) - w_1P_2(z_1), \ \ w_2Q_1(z_2) - w_1Q_2(z_1).
\end{align*}
Then $J \cap \mathbb{K}[w_0,w_1,w_2] \neq \{0\}$.
\end{lemma}

\begin{proof}
Set $p_i = w_i P_1(z_i) - w_{i - 1} P_2(z_{i-1}), q_i = w_i Q_1(z_i) - w_{i - 1} Q_2(z_{i - 1})$, $i = 1, 2$. 
By assumption, the ideal $J$ is generated by $p_1, q_1, p_2, q_2$. 
Next, we construct a nonzero polynomial $R$ in $J \cap \mathbb{K}[w_0,w_1,w_2]$ as follows:
\begin{itemize}
\item[(1)] Set $R_0=\Res_{z_0}\left( \frac{p_1}{w_1},\frac{q_1}{w_1} \right)$ 
 and $R_2=\Res_{z_2}\left( \frac{p_2}{w_1},\frac{q_2}{w_1} \right)$;
\item[(2)] Set $R_1=\Res_{z_1} \left( R_0,R_2 \right)$.
\end{itemize}
By the above construction, we have $R_0 \in \mathbb{K}[\frac{w_0}{w_1},z_1]$, $R_2 \in \mathbb{K}[\frac{w_2}{w_1},z_1]$, and $R_1 \in \mathbb{K}[\frac{w_0}{w_1},\frac{w_2}{w_1}]$. 
By \cite[Proposition~5, p.~164]{Cox2015}, we see 
that $R_1$ is a $\K[\frac{w_0}{w_1}, \frac{w_2}{w_1}, z_0, z_1]$-linear combination of polynomials $\frac{p_1}{w_1},\, \frac{q_1}{w_1},\, \frac{p_2}{w_1}$ and $\frac{q_2}{w_1}$ 
,\ie, $$R_1=h_1 \cdot \frac{p_1}{w_1} + h_2 \cdot \frac{q_1}{w_1} + h_3 \cdot \frac{p_2}{w_1} + h_4 \cdot \frac{q_2}{w_1}$$
for some $h_i \in \K[\frac{w_0}{w_1}, \frac{w_2}{w_1}, z_0, z_1]$, $i = 1, 2, 3, 4$.
Then there exists a large enough natural number $N$ 
such that $H_i = w_1^{N-1} \cdot h_i \in \mathbb{K}[w_0, w_1, w_2, z_0, z_1, z_2]$ for $i=1,2,3,4$.
Set $R=w_1^N \cdot R_1$. Then we have
$$R=H_1 \cdot p_1 + H_2 \cdot q_1 + H_3 \cdot p_2 + H_4 \cdot q_2$$
for some $H_i \in \mathbb{K}[w_0, w_1, w_2, z_0, z_1, z_2]$, $i = 1, 2, 3, 4$.
Therefore, we have $R \in J$. On the other hand, since $R_1 \in \mathbb{K}[\frac{w_0}{w_1},\frac{w_2}{w_1}]$ 
and $R=w_1^N \cdot R_1$, we conclude that $R \in \K[w_0, w_1, w_2]$.
Next, we prove that $R \neq 0$.

Suppose that $R=0$. Then we have $R_1 = 0$.
By \cite[Proposition~3, p.163]{Cox2015}, there exists $y_{1,0} \in \overline{\mathbb{K}(w_0,w_1,w_2)}$ such that 
$R_0(w_0,w_1,y_{1,0}) = R_2(w_0,w_2,y_{1,0})=0$.
It follows from Lemma~\ref{lem:2resultant}  that $R_0 \in \mathbb{K}[\frac{w_0}{w_1},y_1] \setminus \{ 0 \}$ 
and $R_2 \in \mathbb{K}[\frac{w_2}{w_1},y_1] \setminus \{ 0 \}$.
Therefore, we see that $y_{1,0} \in \overline{\mathbb{K}(\frac{w_0}{w_1})} \cap \overline{\mathbb{K}(\frac{w_2}{w_1})}$.
Moreover, since $\frac{w_0}{w_1}$ and $\frac{w_2}{w_0}$ are algebraically independent over $\mathbb{K}$,
we have $\overline{\mathbb{K}(\frac{w_0}{w_1})} \cap \overline{\mathbb{K}(\frac{w_2}{w_1})} = \mathbb{K}$.
Thus, we conclude that $y_{1,0} \in \mathbb{K}$.

By Proposition~\ref{prop:resultant_commute_substitution}, we have
\begin{align*}
\Res_{y_0}\left( \frac{p_1(w_0,w_1,y_0,y_{1,0})}{w_1}, 
 \frac{q_1(w_0,w_1,y_0,y_{1,0})}{w_1} \right) = R_0 (w_0,w_1,y_{1,0}) = 0.
\end{align*}
This shows that there exists $y_{0,0} \in \overline{\mathbb{K}(w_0,w_1)}$ such that
\[\frac{p_1(w_0,w_1,y_{0,0},y_{1,0})}{w_1} = \frac{q_1(w_0,w_1,y_{0,0},y_{1,0})}{w_1} = 0,\]
or equivalently,
\begin{equation}\label{eq:PQ2}
P_1(y_{1,0}) = \frac{w_0}{w_1} P_2(y_{0,0}), \quad \text{and} \quad
Q_1(y_{1,0}) = \frac{w_0}{w_1} Q_2(y_{0,0}).
\end{equation}
Since $\gcd(P_2,Q_2)=1$, the polynomials $P_2(y_{0,0})$ and $Q_2(y_{0,0})$ cannot be both zero functions.
Without loss of generality, we can assume that $Q_2(y_{0,0}) \neq 0$.
From~\eqref{eq:PQ2}, we see that $\frac{P_1(y_{1,0})}{Q_1(y_{1,0})} = \frac{P_2(y_{0,0})}{Q_2(y_{0,0})}$.
Therefore, the function $y_{0,0}$ is a zero of the rational function $\frac{P_1(y_{1,0})}{Q_1(y_{1,0})} - \frac{P_2(z_{0})}{Q_2(z_{0})} \in \mathbb{K}(z_0)$.
Since $\K$ is algebraically closed, we conclude that $y_{0,0} \in \mathbb{K}$.
However, since $\frac{w_0}{w_1}$ is transcendental over $\mathbb{K}$, 
the only possibility for $Q_1(y_{1,0}) = \frac{w_0}{w_1} Q_2(y_{0,0})$ to be true is 
that $Q_2(y_{0,0}) = 0$, a contradiction.
\end{proof}

\begin{theorem}\label{THM:nonzero_intersection}
Let $P_i,Q_i$ be polynomials in $\K[z] \setminus \{ 0 \}$ 
such that $\gcd(P_i,Q_i)=1$ and~$\deg \frac{P_i}{Q_i} = n \geq 1$, where $i = 1, 2$. 
Set $I$ to be the ideal of $\mathbb{K}[w_0,w_1,w_2,z_0,z_1,z_2]$ generated by the following polynomials:
\begin{align*}
w_1^n P_1 \left(\frac{z_1}{w_1}\right) - w_0^n P_2 \left(\frac{z_0}{w_0}\right), \ \
w_1^n Q_1 \left(\frac{z_1}{w_1}\right) - w_0^n Q_2 \left(\frac{z_0}{w_0}\right),\\
w_2^n P_1 \left(\frac{z_2}{w_2}\right) - w_1^n P_2 \left(\frac{z_1}{w_1}\right), \ \
w_2^n Q_1 \left(\frac{z_2}{w_2}\right) - w_1^n Q_2 \left(\frac{z_1}{w_1}\right).
\end{align*}
Then $I \cap \mathbb{K}[w_0,w_1,w_2] \neq \{0\}$.
\end{theorem}

\begin{proof}
Set 
\[
 f_i = w_i^n P_1 \left(\frac{z_i}{w_i}\right) - w_{i - 1}^n P_2 \left(\frac{z_{i - 1}}{w_{i - 1}}\right), \ \
g_i = w_i^n Q_1 \left(\frac{z_i}{w_i}\right) - w_{i - 1}^n Q_2 \left(\frac{z_{i - 1}}{w_{i - 1}}\right),
\]
where $i = 1, 2$. Then $I$ is generated by $f_1, g_1, f_2, g_2$. Furthermore, set
\[
p_i = w_i P_1(z_i) - w_{i - 1} P_2(z_{i-1}), \ \ q_i = w_i Q_1(z_i) - w_{i - 1} Q_2(z_{i - 1}), 
\]
where $i = 1, 2$. 

Set $\E = \overline{\mathbb{K}(w_0,w_1,w_2)}$. 
Consider the following two algebraic systems over the ground field $\E$,
\begin{equation}\label{sys:f}
f_1=g_1=f_2=g_2=0,
\end{equation}
and
\begin{equation}\label{sys:p}
p_1=q_1=p_2=q_2=0,
\end{equation}
where $z_0,z_1,z_2$ are indeterminates.

First, we prove that the above two algebraic systems have the same 
consistent property, \ie,~\eqref{sys:f} has solutions in $\E^3$ 
if and only if~\eqref{sys:p} has solutions in $\E^3$.
Assume that $(y_{0}, y_{1}, y_{2}) \in \E^3$ is a solution of system~\eqref{sys:f}, then it is clear to see that 
\[\left( 
\frac{y_{0}(\sqrt[n]{w_0},\sqrt[n]{w_1},\sqrt[n]{w_2})}{\sqrt[n]{w_0}},
\frac{y_{1}(\sqrt[n]{w_0},\sqrt[n]{w_1},\sqrt[n]{w_2})}{\sqrt[n]{w_1}},
\frac{y_{2}(\sqrt[n]{w_0},\sqrt[n]{w_1},\sqrt[n]{w_2})}{\sqrt[n]{w_2}} 
\right) \in \E^3\]
is a solution of system~\eqref{sys:p}.
Conversely, if $(y_{0}, y_{1}, y_{2}) \in \E^3$ is a solution of system~\eqref{sys:p}, then
\[\left(
w_0 \cdot y_{0}(w_0^n,w_1^n,w_2^n),
w_1 \cdot y_{1}(w_0^n,w_1^n,w_2^n),
w_2 \cdot y_{2}(w_0^n,w_1^n,w_2^n)
\right) \in \E^3\]
is a solution of system~\eqref{sys:f}.

Due to Lemma~\ref{lem:4resultant}, 
we can derive a consequence for system~\eqref{sys:p} of the form 
$$p(w_0,w_1,w_2)=0$$ 
for some nonzero polynomial $p \in \mathbb{K}[w_0,w_1,w_2]$. 
Since $p$ is a nonzero element in the ground field $\E$, 
the above equation is equivalent to $1 = 0$.  
Therefore, system~\eqref{sys:p} has no solution in $\E^3$.
By the consistent property, so is system~\eqref{sys:f}.
Due to the weak version of Hilbert Nullstellensatz \cite[Theorem~5.4,~p.~33]{Matsumura1970}, the ideal of $\mathbb{K}(w_0,w_1,w_2)[z_0,z_1,z_2]$ generated by $I$ contains $1$.
In other words, we have
\begin{equation}\label{eq:Null}
1=h_1 f_1 + h_2 g_1 + h_3 f_2+ h_4 g_2
\end{equation}
for some polynomials $h_1,h_2,h_3,h_4$ in $\mathbb{K}(w_0,w_1,w_2)[z_0,z_1,z_2]$.
Let $h \in \mathbb{K}[w_0,w_1,w_2]$ be the common denominator of the coefficients (in $\mathbb{K}(w_0,w_1,w_2)$) of $h_1,h_2,h_3,h_4$.
Multiplying both sides of \eqref{eq:Null} by $h$, 
we see that $h \in I$.
Hence, we conclude that $h$ is a nonzero polynomial in $I \cap \mathbb{K}[w_0,w_1,w_2]$.
\end{proof}

Note that if $P_i, Q_i$ are polynomials in $\K[z] \setminus \{ 0 \}$  of degrees at most $n \geq 1$ such that  $\gcd(P_i,Q_i)=1$  for $i = 1, 2$, 
then the  claim in the above theorem still holds. 
However, the condition that $\deg \frac{P_i}{Q_i} = n \geq 1$  is necessary for the next corollary, where $i = 1, 2$.

\begin{corollary} \label{COR:nonzero_intersection}
Let $I$ be the ideal defined in Theorem~\ref{THM:nonzero_intersection}. 
Then $I \cap \mathbb{K}[z_0,z_1,z_2] \neq \{0\}$.
\end{corollary}
\begin{proof}
For $i \in \{1, 2 \}$, we denote new polynomials $\hat{P}, \hat{Q}$ in $\mathbb{K}[z] \setminus \{0\}$ as follows: 
$$\hat{P}_i(z) = z^n P_i \left( \frac{1}{z} \right) \qquad \text{and} \qquad \hat{Q}_i(z) = z^n Q_i \left( \frac{1}{z} \right).$$
Since $\deg (\frac{P_i}{Q_i}) = n$ and $\gcd(P_i, Q_i)=1$, we have $\deg (\frac{\hat{P}_i}{\hat{Q}_i}) = n$ and $\gcd(\hat{P}_i, \hat{Q}_i) = 1$. 
The four generators of $I$ can be rewritten in terms of $\hat{P}_i, \hat{Q}_i$ as
\begin{align*}
z_1^n \hat{P}_1 \left(\frac{w_1}{z_1}\right) - z_0^n \hat{P}_2 \left(\frac{w_0}{z_0}\right),\ \ 
z_1^n \hat{Q}_1 \left(\frac{w_1}{z_1}\right) - z_0^n \hat{Q}_2 \left(\frac{w_0}{z_0}\right),\\
z_2^n \hat{P}_1 \left(\frac{w_2}{z_2}\right) - z_1^n \hat{P}_2 \left(\frac{w_1}{z_1}\right), \ \
z_2^n \hat{Q}_1 \left(\frac{w_2}{z_2}\right) - z_1^n \hat{Q}_2 \left(\frac{w_1}{z_1}\right).
\end{align*}
We conclude from Theorem~\ref{THM:nonzero_intersection} that $I \cap \mathbb{K}[z_0,z_1,z_2] \neq \{0\}$.
\end{proof}

\begin{corollary} \label{COR:elimination}
Algorithm~\ref{alg:system_to_single_DEs} is correct. 
\end{corollary}
\begin{proof}
Assume that $(A(x), B(x)) \in \K[x]^2$ is a solution of ~\eqref{sys:system_polynomial} with $\gcd(A, B) = 1$.  
Set
$$
\begin{array}{cccc}
\phi: & \K[w_0, w_1, w_2, z_0, z_1, z_2] & \to & \K[x]\\
& F(w_0, w_1, w_2, z_0, z_1, z_2) & \mapsto & F(B(x), B(x + 1), B(x + 2), A(x), A(x + 1), A(x + 1)).
\end{array}
$$
It is clear to see that $\phi$ is a ring homomorphism. 
Applying the usual shift operator to system~\eqref{sys:system_polynomial}, we obtain
\begin{equation*}\label{sys:system_shift}
\left \{ 
\begin{aligned}
&\tilde{P}_1(A(x+1),B(x+1)) = \tilde{P}_2(A(x),B(x)),\\
&\tilde{Q}_1(A(x+1),B(x+1)) = \tilde{Q}_2(A(x),B(x)),\\
&\tilde{P}_1(A(x+2),B(x+2)) = \tilde{P}_2(A(x+1),B(x+1)),\\
&\tilde{Q}_1(A(x+2),B(x+2)) = \tilde{Q}_2(A(x+1),B(x+1)).
\end{aligned}
\right. 
\end{equation*}
The above system implies that $I \subseteq \ker(\phi)$. 
By Theorem~\ref{THM:nonzero_intersection} and Corollary~\ref{COR:nonzero_intersection}, 
there exist nonzero elements $F_A \in I \cap \K[z_0, z_1, z_2]$ and $F_B \in I \cap \K[w_0, w_1, w_2]$. 
Therefore, we conclude that $\phi(F_A) = 0$ and $\phi(F_B) = 0$.
\end{proof}

From experiments, we observe that resultant computation is much more efficient than that of Gr\"{o}bner bases in Step $1$ of Algorithm~\ref{alg:system_to_single_DEs}. 
Since generators of the ideal~$I$ of Algorithm~\ref{alg:system_to_single_DEs} are homogenous, the output difference equations are also homogenous. 
We will present  an algorithm for computing polynomial solutions of second-order autonomous homogenous {\AODEs} in the next subsection. 

%
%
%

\subsection{Polynomial solutions of autonomous second-order homogenous {\AODEs}}\label{SUBSEC:polysol}


In this subsection, we consider the following problem:

\begin{problem}\label{prob:second_order_odes}
Let $F \in \mathbb{K}[y,z,w]$ be a trivariate homogeneous polynomial. 
Find all polynomial solutions of the difference equation $F(y(x),y(x+1),y(x+2))=0.$
\end{problem}


An answer for the above problem is the last key for deriving a complete algorithm for finding rational solutions of an autonomous first-order \AODE. 
In literature, there exists an algorithm for computing polynomial solutions of absolutely irreducible autonomous first-order {\AODEs}~\cite[Section 3]{FengGao2008} and also a degree bound~\cite{Eekelen2014} for 
polynomial solutions of general {\AODEs} under a certain sufficient condition. 
However, none of them can give a direct answer for Problem~\ref{prob:second_order_odes}.
We will give a degree bound for polynomial solutions of autonomous second-order homogeneous {\AODEs}, 
and thus it can be used to derive an algorithm for computing the corresponding polynomial solutions.
The approach is a difference analog of that in~\cite{VoZhang18, Vo2018Computation}.

Assume that $F \in \K[y, z, w]$ is a homogenous polynomial of total degree $D$. 
Consider the following second-order \AODE: 
\begin{equation} \label{EQ:polysolsec1}
 F(y(x), y(x + 1), y(x + 2)) = 0.
\end{equation}

Let $\Delta y(x) = y(x + 1) - y(x)$. Then we have
\begin{align*}
 y(x + 1) & = \Delta y(x) + y(x)\\
 y(x + 2) & = \Delta^2 y(x) + 2 \Delta y(x) + y(x).
\end{align*}
Substituting the above two formulae into~\eqref{EQ:polysolsec1}, we obtain
the algebraic equation: 
\begin{equation} \label{EQ:polysolsec2}
 \tilde{F}(y(x), \Delta y(x), \Delta^2 y(x)) = 0,
\end{equation}
where 
$$\tilde{F}(y,z,w) = F(y,y+z,y+2z+w)$$ 
is also a homogenous polynomial of total degree $D$ in $\K[y, z, w]$.

For each $I = (i_1, i_2, i_3) \in \N^3$, 
we define $|| I || = i_1 + i_2 + i_3$. 
Now we may write 
\begin{equation} \label{EQ:polysolsec3}
\tilde{F} = \sum_{||I|| = D} c_I y^{i_1} z^{i_2} w^{i_3},
\end{equation}
where $c_I \in \K$. Set
\begin{align*}
 \mathcal{E}(\tilde{F}) & = \{ I \in \N^3 \,|\, c_{I} \neq 0 \},\\
 m(\tilde{F}) & = \min\{ i_2 + 2 i_3 \,|\, I \in \mathcal{E}(\tilde{F}) \}, \\
 \mathcal{M}(\tilde{F}) & = \{ I \in \mathcal{E}(\tilde{P}) \,|\, i_2 + 2 i_3 = m(\tilde{F})\}, \\
  \mathcal{P}_{\tilde{F}}(t) & = \sum_{I \in \mathcal{M}(\tilde{F})} c_I t^{i_2} [t(t-1)]^{i_3}. 
\end{align*}
We call $\mathcal{P}_{\tilde{F}}(t)$ the \emph{indicial polynomial} of $\tilde{F}$ (at infinity). 

\begin{proposition} \label{PROP:nonzeroindicialpol}
Let $\mathcal{P}_{\tilde{F}}(t)$ be the indicial polynomial of $\tilde{F}$ at infinity. 
Then $\mathcal{P}_{\tilde{F}}(t) \neq 0$.
\end{proposition}
\begin{proof}
Set $f(t) = \frac{\mathcal{P}_{\tilde{F}}(t)}{t^{m(\tilde{F})}}$. Then 
\[
 f(t) = \sum_{I \in \mathcal{M}(\tilde{F})} c_I \left( \frac{t - 1}{t} \right)^{i_3}. 
\]
Let $T = \frac{t - 1}{t}$. Then $f(t)$ is the evaluation of the nonzero
univariate polynomial 
$$\sum_{I \in \mathcal{M}(\tilde{F})} c_I x^{i_3} \ \text{ at } \ T.$$ 
Since $T$ is transcendental over $\K$, we conclude that $f(t)$ is nonzero. 
Thus, $\mathcal{P}_{\tilde{F}}(t) \neq 0$.
\end{proof}

Assume that $p(x) = \sum_{i = 0}^{d} a_i x^i \in \K(c)[x]$ 
is a nonzero polynomial solution of~\eqref{EQ:polysolsec2}, 
where $c$ is transcendental over $\K(x)$, $d$ and $a_i$'s are unkown. Then 
\begin{align*}
 \Delta p(x) & =  a_d \cdot d \cdot x^{d - 1} +  
 \  \text{ lower terms in } x, \\
 \Delta^2 p(x) & = a_d \cdot d \cdot (d - 1) \cdot x^{d - 2} +  
 \  \text{ lower terms in } x.
\end{align*}
Thus, for each $I = (i_1, i_2, i_3) \in \N^3$ with $||I|| = D$, 
we have 
\begin{align}
p^{i_1} (\Delta p)^{i_2} (\Delta^2 p)^{i_3} 
& = (a_d x^d)^{i_1} (a_d d x^{d-1})^{i_2} (a_d d (d - 1) x^{d-2})^{i_3} 
+ \text{ lower terms in } x \nonumber \\
 & = a_d^{i_1 + i_2 + i_3} d^{i_2} [d(d-1)]^{i_3} x^{d (i_1 + i_2 + i_3) - (i_2 + 2 i_3)} + \text{ lower terms in } x \nonumber \\ 
 & = a_d^D d^{i_2} [d(d-1)]^{i_3} x^{d D - (i_2 + 2 i_3)} + 
 \text{ lower terms in } x. \label{EQ:leadterm} 
\end{align}

Based on the above argument, we have the following proposition. 

\begin{proposition} \label{PROP:polysoldegbound}
Let $p(x)$ be a nonzero polynomial solution of~\eqref{EQ:polysolsec2} with degree $d$. Then $\mathcal{P}_{\tilde{F}}(d) = 0$.
\end{proposition}
\begin{proof}
Since $p(x)$ is a nonzero polynomial solution of~\eqref{EQ:polysolsec2} with degree $d$, we have that 
\[
 [x^{d D - m(\tilde{F})}] \left( \tilde{F}(p(x), \Delta p(x), \Delta^2 p(x)) \right) = 0. 
\]
By~\eqref{EQ:leadterm}, the above equation is equivalent to 
 $a_d^D \cdot \mathcal{P}_{\tilde{F}}(d) = 0$. 
On account of $a_d \neq 0$, we conclude that $\mathcal{P}_{\tilde{F}}(d) = 0$.
\end{proof}

By Proposition~\ref{PROP:nonzeroindicialpol} and the above one, we can compute a degree bound of polynomial solutions of $\tilde{F}$
by computing non-negative integers solutions of the indicial polynomial at infinity. 
Afterwards, we may compute the polynomial solutions by making an ansatz and then solving the corresponding 
algebraic equations by using Gr\"{o}bner bases.

\section{Algorithms for computing rational solutions of autonomous first-order \AODEs} \label{SECT:algo}

Based on the results of the previous section, we summarize an algorithm for 
determining rational solutions of an autonomous first-order \AODE.

\begin{algo} \label{ALGO:firstseparable}
Given a separable difference equation $\frac{P_1(y(x + 1))}{Q_1(y(x+ 1))} = \frac{P_2(y(x))}{Q_2(y(x))}$ 
with $\gcd(P_i,Q_i) = 1$ and $\deg \frac{P_1}{Q_1} = \deg \frac{P_2}{Q_2} \geq 1$, $i = 1, 2$, 
compute a bound $N$ for the degrees of its rational solutions.

\begin{enumerate}
	\item Set $N=0$. If $\frac{P_1(z)}{Q_2(z)} = \frac{P_1(z)}{Q_2(z)}$, then output $N$. Otherwise, go to the next step. 
	\item Compute the set of constant candidates $\mathcal{C}$ of the given separable 
difference equation by Definition~\ref{def:C}. 
	\item Let $c_1, c_2, \ldots, c_m$ be nonzero elements in $\mathcal{C}$. Let 
	\begin{equation*}
	\tilde{P}_j(z,w) = w^n P_j \left(\frac{z}{w}\right), \quad \text{ and } \quad \tilde{Q}_j(z,w) = w^n Q_j \left(\frac{z}{w}\right), \qquad j = 1, 2.
	\end{equation*}
	For $i$ from $1$ to $m$ do
	\begin{enumerate}
		\item Consider the difference system
		\begin{equation} \label{EQ:c_i}
		\left \{ 
		\begin{aligned}
		\tilde{P}_1(A(x+1),B(x+1)) = c_i \cdot \tilde{P}_2(A(x),B(x)),\\
		\tilde{Q}_1(A(x+1),B(x+1)) = c_i \cdot \tilde{Q}_2(A(x),B(x)).
		\end{aligned}
		\right.
		\end{equation}
		where $A,B$ are unknown functions.
		Derive the following two nonzero autonomous second-order {\AODEs} for $A(x)$ and B(x) from the above equations by using Algorithm~\ref{alg:system_to_single_DEs}:
		\[
		F_{i, A}(A(x), A(x +1), A(x+2)) = 0, \quad \text{ and } \quad 			F_{i, B}(B(x), B(x + 1), B(x+2)) =0,  
		\]
		where $F_{i, A}$ and  $F_{i, B}$ are homogeneous polynomials in  $\K[y, z, w] \setminus \{ 0 \}$.
		\item Determine the indicial polynomials $\mathcal{P}_{F_{i,A}}$ and $\mathcal{P}_{F_{i,B}}$ of $F_{i,A}$ and $F_{i,B}$, respectively. 
                       Let 
\begin{align*}
D_{i, A} & = \{ \text{non-negative integer solutions of } \mathcal{P}_{F_{i,A}}(t)\}, \\
D_{i, B} & = \{ \text{non-negative integer solutions of } \mathcal{P}_{F_{i,B}}(t)\}.
\end{align*}
		\item Set $N =\max \{\{N \} \cup D_{i, A} \cup D_{i, B}\}$.
	\end{enumerate}
\item Return $N$. 
\end{enumerate}

\end{algo}

The termination of the above algorithm is evident.
The correctness is a consequence of  Proposition~\ref{prop:f=f}, Theorem~\ref{THM:exist_c_C}, the correctness of Algorithm~\ref{alg:system_to_single_DEs} and Proposition~\ref{PROP:polysoldegbound}.

To avoid triviality, we only consider non-constant rational solutions below.

\begin{algo} \label{ALGO:firstautonomous}
Given an irreducible autonomous first-order {\AODE} $F(y(x), y(x + 1)) = 0$, compute a non-constant rational solution or return NULL.  
\begin{enumerate}
\item If $\deg_y(F) \neq \deg_z(F)$, then output NULL. Otherwise, go to the next step. 
\item Compute the genus $g$ of the corresponding curve $\mathcal{C}_F$ defined by $F(y,z)=0$. If $g \neq 0$, then output NULL. Otherwise, go to the next step.
\item By using \cite[Algorithm~1]{Vo2018Deciding}, determine an optimal parametrization for $\mathcal{C}_F$, say $\mathcal{P}(t)=(p_1(t), p_2(t))$.
\item Apply Algorithm~\ref{ALGO:firstseparable} to compute a bound $N$ for the degrees of rational solutions of the associated separable difference equation $p_1(y(x+1))=p_2(y(x))$.  
\item Set $M=N \cdot \deg_t p_1$. Use \cite[Algorithm~4.16]{FengGao2008} to determine a non-constant rational solution of the given {\AODE} whose degree at most $M$. Return the rational solution if there is any. Otherwise return NULL.
\end{enumerate}
\end{algo}

The termination of the above algorithm is clear.

\begin{theorem}
Algorithm~\ref{ALGO:firstautonomous} is correct.
\end{theorem}

\begin{proof}
Assume that the difference equation $F(y(x),y(x+1))=0$ admits a rational solution. 
Then it has a strong rational general solution. By Proposition~\ref{prop:ReducedEquation}, we have $\deg_yF = \deg_zF$. 
By Theorem~\ref{thm:StrongSolution}, the genus of the corresponding curve $\mathcal{C}_F$ is of genus zero.
In this case, $\mathcal{C}_F$ admits an optimal parametrization $\mathcal{P}(t)=(p_1(t),p_2(t)) \in \mathbb{K}(t) \times \mathbb{K}(t)$ and the given difference equation has an associated separable difference equation.
By Theorem~\ref{thm:OneToOne}, we see that if $N$ is a bound for the degrees of rational solutions of the associated separable difference equation, then $N \cdot \deg_tp_1$ is a bound for the degrees of rational solutions of the given difference equation. 
The correctness of step~$5$ follows from that of~\cite[Algorithm~4.16]{FengGao2008}.
\end{proof}



\begin{example} \label{EX:FengGaoHuang2008}
Consider the following irreducible autonomous first-order \AODE:
\begin{equation} \label{EQ:FengGaoHuang2008}
 F = (12 y(x) + 49) y(x+1)^2 - (12 y^2 + 62 y + 56) y(x + 1) + y(x)^2 +  8 y(x) + 16 = 0. 
\end{equation}
It is clear to see that $\deg_y(F) = \deg_z(F) = 2$. The corresponding algebraic curve is of genus zero and it has an optimal parametrization
$$\mathcal{P}(t)=(p_1(t), p_2(t) ) = \left( \frac{9 t^2 - 12 t + 4}{12 t},\frac{9 t^2 + 36 t + 4}{12 (t + 4)} \right).$$  
Using the above parametrization, we can derive the following associated separable difference equation of~\eqref{EQ:FengGaoHuang2008}:
\begin{equation*} 
\frac{9 y(x + 1)^2 - 12 y(x+1) + 4}{y(x+1)} = \frac{9 y(x)^2 + 36 y(x) + 4}{y(x) + 4},
\end{equation*} 
where $P_1(z) = 9 z^2 - 12 z + 4, Q_1(z) = z, P_2(z) = 9 z^2 + 36 z + 4$ and $Q_2(z) = z + 4$. 

It is clear that  $\frac{P_1(z)}{Q_2(z)} \neq \frac{P_1(z)}{Q_2(z)}$. So we can skip step 1 of Algorithm~\ref{ALGO:firstseparable}. 
By computation, we find that the set of candidate constants $\mathcal{C} = \{7 + 4 \sqrt{3}, 7 - 4 \sqrt{3}, 0, 1\}$. 
Using other steps of Algorithm~\ref{ALGO:firstseparable}, we see that the degree bound for rational solutions of the associated separable difference equation is 2.
Thus, the degrees of rational solutions of the given difference equation are bounded by 4.
By applying \cite[Algorithm~4.1]{FengGao2008}, we can determine a rational solution, say
\[
y(x) = \frac{(1 - 4 x + 2 x^2)^2}
{2x (1 - 3 x + 2 x^2)}.
\]


\end{example}

\section*{Acknowledgement}
We thank Georg Grasegger and Matteo Gallet for valuable suggestions and comments on revising our paper.


\bibliographystyle{plain}

\begin{thebibliography}{10}

\bibitem{Abramov1989}
S.~A.~Abramov.
\newblock Rational solutions of linear differential and difference equations
  with polynomial coefficients.
\newblock {\em USSR Computational Mathematics and Mathematical Physics}, 29(6):7--12, 1989.

\bibitem{Abramov1995}
S.~A.~Abramov.
\newblock Rational solutions of linear difference and $q$-difference equations
  with polynomial coefficients.
\newblock In {\em Proceedings of the 1995 International Symposium on Symbolic
  and Algebraic Computation}, ISSAC '95, pages 285--289, New York, NY, USA,
  1995. ACM.

\bibitem{Abramov1998}
S.~A.~Abramov.
\newblock Rational solutions of first order difference systems.
\newblock In {\em Proceedings of the 1998 International Symposium on Symbolic
  and Algebraic Computation}, ISSAC '98, pages 124--131, New York, NY, USA,
  1998. ACM.

\bibitem{AbramovBronstein1995}
S.~A.~Abramov, M.~Bronstein, and M.~Petkov\u{s}ek.
\newblock On polynomial solutions of linear operators equations.
\newblock In {\em Proceedings of the 1995 International Symposium on Symbolic
  and Algebraic Computation}, ISSAC '95, pages 290--295, New York, NY, USA,
  1995. ACM.

\bibitem{AbramovPaule1998}
S.~A.~Abramov, P.~Paule, and M.~Petkov\u{s}ek.
\newblock $q$-hypergeometric solutions of $q$-difference equations.
\newblock {\em Discrete Mathematics}, 100:3--22, 1998.

\bibitem{ArocaCanoFengGao}
J.~M.~Aroca, J.~Cano, R.~Feng, and X.~S.~Gao.
\newblock Algebraic general solutions of algebraic ordinary differential
  equations.
\newblock In {\em Proceedings of the 2005 International Symposium on Symbolic
  and Algebraic Computation}, ISSAC '05, pages 29--36, New York, NY, USA, 2005.
  ACM.

\bibitem{Binder1996}
F.~Binder.
\newblock Fast computations in the lattice of polynomial rational function
  fields.
\newblock In {\em Proceedings of the 1996 International Symposium on Symbolic
  and Algebraic Computation}, ISSAC '96, pages 43--48, New York, NY, USA, 1996.
  ACM.

\bibitem{Koepf1999}
H.~B\"{o}ing and W.~Koepf.
\newblock Algorithms for $q$-hypergeometric summation in computer algebra.
\newblock {\em Journal of Symbolic Computation}, 28(6):777--799, 1999.

\bibitem{Bronstein2000}
M.~Bronstein.
\newblock On solutions of linear ordinary difference equations in their
  coefficient field.
\newblock {\em Journal of Symbolic Computation}, 29(6):841--877, 2000.

\bibitem{BronsteinPetkovsek1996}
M.~Bronstein and M.~Petkov{\v{s}}ek.
\newblock An introduction to pseudo-linear algebra.
\newblock {\em Theoretical Computer Science}, 157:3--33, 1996.

\bibitem{Cohn1965}
R.~M.~Cohn.
\newblock {\em Difference Algebra}, volume~17 of {\em Tracts in Mathematics}.
\newblock Interscience, New York, 1965.

\bibitem{Cox2015}
D.~Cox, J.~Little, and D.~O'Shea.
\newblock {\em Ideals, Varieties, and Algorithms}.
\newblock Springer International Publishing, 4th edition, 2015.

\bibitem{Elaydi2004}
S.~Elaydi.
\newblock {\em An Introduction to Difference Equations}.
\newblock Undergraduate Texts in Mathematics. Springer New York, 3th edition,
  2004.

\bibitem{FengGao}
R.~Feng and X.~S.~Gao.
\newblock Rational general solutions of algebraic ordinary differential
  equations.
\newblock In {\em Proceedings of the 2004 International Symposium on Symbolic
  and Algebraic Computation}, ISSAC '04, pages 155--162, New York, NY, USA,
  2004. ACM.

\bibitem{FengGao06}
R.~Feng and X.~S.~Gao.
\newblock A polynomial time algorithm for finding rational general solutions of
  first order autonomous odes.
\newblock {\em Journal of Symbolic Computation}, 41(7):739--762, 2006.

\bibitem{FengGao2008}
R.~Feng, X.~S.~Gao, and Z.~Huang.
\newblock Rational solutions of ordinary difference equations.
\newblock {\em Journal of Symbolic Computation}, 43(10):746 -- 763, 2008.

\bibitem{Gosper1978}
R.~W.~Gosper.
\newblock Decision procedure for indefinite hypergeometric summation.
\newblock {\em Proceedings of the National Academy of Sciences of the United States of America}, 75(1):40--42, 1978.

\bibitem{Hartshorne}
R.~Hartshorne.
\newblock {\em Algebraic Geometry}, volume~52 of {\em Graduate Texts in
  Mathematics}.
\newblock Springer-Verlag, New York Berlin Heidelberg, 1977.

\bibitem{HendriksSinger1999}
P.~A.~Hendriks and M.~F.~Singer.
\newblock Solving difference equations in finite terms.
\newblock {\em Journal of Symbolic Computation}, 27(3):239--259, 1999.

\bibitem{Karr1981}
M.~Karr.
\newblock Summation in finite terms.
\newblock {\em Journal of the Association for Computing Machinery}, 28(2):305--350, 1981.

\bibitem{Karr1985}
M.~Karr.
\newblock Theory of summation in finite terms.
\newblock {\em Journal of Symbolic Computation}, 1(3):303--315, 1985.

\bibitem{Kauers2006}
M.~Kauers and C.~Schneider.
\newblock Indefinite summation with unspecified summands.
\newblock {\em Discrete Mathematics}, 306(17):2073--2083, 2006.

\bibitem{Koepf1995}
W.~Koepf.
\newblock Algorithms for $m$-fold hypergeometric summation.
\newblock {\em Journal of Symbolic Computation}, 20(3):399--417, 1995.

\bibitem{KoutschanThesis}
C.~Koutschan.
\newblock {\em Advanced applications of the holonomic systems approach}.
\newblock PhD thesis, Johannes Kepler University Linz, 2009.

\bibitem{Matsumura1970}
H.~Matsumura.
\newblock {\em Commutative Algebra}, volume 120.
\newblock WA Benjamin New York, 1970.

\bibitem{Paule1997}
P.~Paule and A.~Riese.
\newblock A Mathematica $q$-analogue of Zeilberger's algorithm based on an
  algebraically motivated approach to $q$-hypergeometric telescoping.
\newblock {\em Special Functions, $q$-series, and Related Topics}, 14:179--210,
  1997.

\bibitem{Paule1995b}
P.~Paule.
\newblock Greatest factorial factorization and symbolic summation.
\newblock {\em Journal of Symbolic Computation}, 20(3):235--268, 1995.

\bibitem{Paule1995}
P.~Paule and M.~Schorn.
\newblock A {M}athematica version of {Z}eilberger's algorithm for proving
  binomial coefficient identities.
\newblock {\em Journal of Symbolic Computation}, 20(5-6):673--698, 1995.
\newblock Symbolic computation in combinatorics $\Delta{\sb{1}}$ (Ithaca, NY,
  1993).

\bibitem{Petkovsek1992}
M.~Petkov{\v{s}}ek.
\newblock Hypergeometric solutions of linear recurrences with polynomial
  coefficients.
\newblock {\em Journal of Symbolic Computation}, 14(2-3):243--264, 1992.

\bibitem{PWZbook1996}
M.~Petkov{\v{s}}ek, H.~S.~Wilf, and D.~Zeilberger.
\newblock {\em {$A=B$}}.
\newblock A K Peters Ltd., Wellesley, MA, 1996.
\newblock With a foreword by Donald E. Knuth.

\bibitem{Ritt1932}
J.~F.~Ritt.
\newblock {\em Differential equations from the algebraic standpoint},
  volume~14.
\newblock American Mathematical Soc., 1932.

\bibitem{Schneider2005}
C.~Schneider.
\newblock Solving parameterized linear difference equations in terms of
  indefinite nested sums and products.
\newblock {\em Journal of Difference Equations and Applications},
  11(9):799--821, 2005.

\bibitem{Sendra2008}
{J.\,R.} Sendra, F.~Winkler, and S.~P\'erez-D\'iaz.
\newblock {\em {Rational Algebraic Curves, A Computer Algebra Approach}},
  volume~22 of {\em Algorithms and Computation in Mathematics}.
\newblock Springer-Verlag, Berlin Heidelberg, 2008.

\bibitem{Eekelen2014}
O.~Shkaravska and M.~van Eekelen.
\newblock Univariate polynomial solutions of algebraic difference equations.
\newblock {\em Journal of Symbolic Computation}, 60:15--28, 2014.

\bibitem{Eekelen2018}
O.~Shkaravska and M.~van Eekelen.
\newblock Polynomial solutions of algebraic difference equations and
  homogeneous symmetric polynomials.
\newblock 2018.

\bibitem{vanHoeij1998}
M.~van Hoeij.
\newblock Rational solutions of linear difference equations.
\newblock In {\em Proceedings of the 1998 International Symposium on Symbolic
  and Algebraic Computation}, ISSAC '98, pages 120--123, New York, NY, USA,
  1998. ACM.

\bibitem{vanHoeij1999}
M.~van Hoeij.
\newblock Finite singularities and hypergeometric solutions of linear
  difference equations.
\newblock {\em Journal of Pure and Applied Algebra}, 139:109--131, 1999.

\bibitem{Hoeij2006Solving}
M.~van Hoeij and J.~Cremona.
\newblock Solving conics over function fields.
\newblock {\em Journal de Th\'{e}orie des Nombres de Bordeaux}, 18:595--606, 2006.

\bibitem{VoZhang18}
T.~N.~Vo and Y.~Zhang.
\newblock {Rational solutions of high-order algebraic ordinary differential
  equations}.
\newblock {\em arXiv: 1709.04174}, pages 1--15, 2018.

\bibitem{Vo2018Deciding}
T.~N. Vo, G.~Grasegger, and F.~Winkler.
\newblock Deciding the existence of rational general solutions for first-order
  algebraic odes.
\newblock {\em Journal of Symbolic Computation}, pages~--, 2017.
\newblock To appear.

\bibitem{Vo2018Computation}
T.~N.~Vo, G.~Grasegger, and F.~Winkler.
\newblock Computation of all rational solutions of first-order algebraic odes.
\newblock {\em Advances in Applied Mathematics}, 98:1--24, 2018.

\bibitem{Wolfram2000}
D.~A.~Wolfram.
\newblock A formula for the general solutions of a constant-coefficient
  difference equations.
\newblock {\em Journal of Symbolic Computation}, 29(1):79--82, 2000.

\end{thebibliography}

\end{document}